\documentclass[ 11pt, english, a4paper]{article}

\usepackage{amsmath,amsfonts}
\usepackage{amssymb}

\usepackage{graphicx}

\usepackage{dcolumn}
\usepackage[dvipsnames,svgnames,table]{xcolor}

\usepackage[sort]{natbib}

\usepackage[T1]{fontenc}
\usepackage[a4paper]{geometry}
\usepackage{color}

\usepackage{setspace}
\usepackage{babel}

\usepackage{mathtools}
\usepackage{booktabs,caption}
\usepackage[flushleft]{threeparttable}
\usepackage{mathptmx}
\normalfont

\usepackage{dsfont}
 \usepackage{amsmath,amsfonts}
\usepackage{amssymb}

\DeclareMathSymbol{\niceE}{\mathalpha}{AMSb}{"45}
\DeclareMathOperator*{\plim}{plim}

\newtheorem{theorem}{Theorem}[section]
\newtheorem{lemma}[theorem]{Lemma}

\newtheorem{proposition}[theorem]{Proposition}
\newtheorem{corollary}[theorem]{Corollary}

\newenvironment{proof}[1][Proof]{\begin{trivlist}
		\item[\hskip \labelsep {\bfseries #1}]}{\end{trivlist}}
\newenvironment{definition}[1][Definition]{\begin{trivlist}
		\item[\hskip \labelsep {\bfseries #1}]}{\end{trivlist}}

\newenvironment{remark}[1][Remark]{\begin{trivlist}
		\item[\hskip \labelsep {\bfseries #1}]}{\end{trivlist}}
\newenvironment{assumption}[1][Assumption]{\begin{trivlist}
		\item[\hskip \labelsep {\bfseries #1}]}{\end{trivlist}}
\usepackage{listings}
\definecolor{darkblue}{rgb}{0.0, 0.0, 0.55}
\definecolor{denim}{rgb}{0.08, 0.38, 0.74}
\definecolor{darkmidnightblue}{rgb}{0.0, 0.2, 0.4}
\begin{document}

\title{A Synthetic Instrumental Variable Method: Using the Dual Tendency Condition for Coplanar Instruments\thanks{The authors thank Arthur Lewbel, Richard Butler, Jeff Wooldridge,  Guido Imbens, and  Xinrui Catherine Yu for valuable comments on the paper. }
}
\author{Ratbek Dzhumashev\thanks{Department of Economics, Monash University,
       \texttt{Ratbek.Dzhumashev@monash.edu}}  and Ainura Tursunalieva\thanks{
  Data61, CSIRO
                  \texttt{ainura.tursunalieva@data61.csiro.au}}
         }
\maketitle  

\thispagestyle{empty}

 \renewcommand{\MakeUppercase}[1]{\color{black}\textrm{#1}}
\pagenumbering{roman} 

\pagenumbering{arabic}

\singlespacing
\begin{abstract}

Traditional instrumental variable (IV) methods often struggle with weak or invalid instruments and rely heavily on external data. We introduce a Synthetic Instrumental Variable (SIV) approach that constructs valid instruments using only existing data. Our method leverages a data-driven dual tendency (DT) condition to identify valid instruments without requiring external variables. SIV is robust to heteroscedasticity and can determine the true sign of the correlation between endogenous regressors and errors--an assumption typically imposed in empirical work. Through simulations and real-world applications, we show that SIV improves causal inference by mitigating common IV limitations and reducing dependence on scarce instruments. This approach has broad implications for economics, epidemiology, and policy evaluation.
\end{abstract}



\onehalfspacing

\section{Introduction}

%
Endogeneity persistently undermines causal inference across economics and the social sciences. When explanatory variables correlate with unobserved factors in the error term, standard regression estimates become biased and inconsistent. This leads to misguided policy conclusions, distorted theoretical interpretations, and unreliable empirical insights. Addressing endogeneity is therefore fundamental to credible research and sound policy design.

Instrumental variable (IV) methods have long been the cornerstone of strategies to recover causal effects in the presence of endogeneity. The logic is compelling: find a variable---the instrument---that affects the endogenous regressor but has no direct effect on the outcome except through that regressor. Yet despite substantial methodological advances \citep{imbens24}, implementing IV methods in practice faces three persistent obstacles.

\textit{First, finding valid instruments is difficult.} Few variables simultaneously satisfy the relevance condition (strong correlation with the endogenous regressor) and the exogeneity condition (no correlation with the error term). Researchers often struggle to justify instrument validity, and debates about instrument quality pervade empirical work \citep{angrist, haus}.

\textit{Second, weak instruments create severe problems.} Even when plausible instruments exist, they frequently exhibit only modest correlation with endogenous variables. Weak instruments produce IV estimates with large finite-sample biases---sometimes exceeding OLS biases---and unreliable inference \citep{CHERNOZHUKOV200868, stock}.

\textit{Third, multiple instruments introduce complications.} Using several instruments simultaneously often exacerbates finite-sample distortions and reduces inference reliability, particularly when some instruments are weak \citep{Bound1995, stock, wool}.

These challenges motivate a fundamental question: \textit{Can we construct valid instruments directly from available data, without relying on external variables?}

\subsection{Our Approach: The Synthetic Instrumental Variable (SIV) Method}

We propose the Synthetic Instrumental Variable (SIV) method, which constructs valid instruments from observed data alone, eliminating dependence on external variables. Our approach builds on a simple geometric insight about linear regression that leads to a powerful identification strategy.

Consider the structural equation $\mathbf{y} = \beta\mathbf{x} + \mathbf{u}$.
The three vectors---the outcome ($\mathbf{y}$), the endogenous regressor ($\mathbf{x}$),
and the structural error ($\mathbf{u}$)—lie in the same two-dimensional plane
$\mathcal{W}$ spanned by $\mathbf{y}$ and $\mathbf{x}$. We write 
$\mathcal W:=\operatorname{span}\{\mathbf x,\mathbf y\}$
This \emph{coplanarity} property means that any valid instrument, after projecting
onto this plane, takes the form a linear combination of vectors within $\mathcal{W}$.

Specifically, we show that any valid instrument $\mathbf{z}_0$ can be represented in the form
\[
\mathbf{z}_0 = \mathbf{x} + k\delta_0 \mathbf{r},
\]
where $\mathbf{r}$ is a vector in $\mathcal{W}$ orthogonal to $\mathbf{x}$,
$\delta_0$ is a scalar parameter, and $k \in \{-1,+1\}$ captures the direction of endogeneity:
\[
k = \begin{cases}
-1 & \text{if } \operatorname{cov}(\mathbf{x},\mathbf{u})>0, \\
+1 & \text{if } \operatorname{cov}(\mathbf{x},\mathbf{u})<0.
\end{cases}
\]

This representation reveals that the entire family of potential instruments lies
within the observable plane $\mathcal{W}$.
The challenge reduces to identifying which specific value of $\delta_0$ yields a valid instrument. Figure \ref{fig4} on p. \pageref{pp} illustrates this geometric structure.

\paragraph{The dual tendency condition: linking theory to data.}
To identify the valid instrument $\mathbf{z}_0$, we introduce the \emph{dual tendency} (DT)
condition—a testable criterion that connects the unobservable exogeneity requirement
to observable data characteristics.

The core intuition is straightforward. Think of searching along a one-dimensional path in the regression plane by varying $\delta$. Each value of $\delta$ produces a candidate instrument $\mathbf s(\delta) = \mathbf x + k\delta\mathbf r$. We seek the point where two conditions align---like finding where two independent tests both confirm ``this is the right instrument.
A valid instrument must satisfy two properties simultaneously:
\begin{enumerate}
\item \textbf{Exogeneity:} The instrument is uncorrelated with the structural error,
      $\mathbb{E}[\mathbf{u}\mid \mathbf{z}_0]={0}$.
\item \textbf{First-stage homoscedasticity:} When the first-stage error term is homoscedastic,
      the valid instrument induces a specific moment restriction, $\mathbf{M}(\delta_0)={0}$,
      where $\mathbf{M}(\delta)$ is a function of observable residuals.
\end{enumerate}
The key insight: these two conditions hold together \textit{only at the true instrument}. By searching over candidate instruments for different values of $\delta$, we identify $\delta_0$ as the value satisfying the testable moment condition $\mathbf M(\delta_0) = 0$. At this point, the unobservable exogeneity condition also holds, yielding a valid synthetic instrument:
\[
\mathbf s^\star = \mathbf x + k\delta_0 \mathbf r.
\]

Critically, the method also determines the sign of $\operatorname{cov}(\mathbf x,\mathbf  u)$ from the data.\footnote{In some settings, the sign of $\operatorname{corr}(\mathbf x, \mathbf u)$ may be known \textit{a priori} \citep{DiTraglia2021, moon}. However, since $\mathbf u$ is unobservable, this assumption is rarely verifiable. Our method infers the sign directly from observable moment conditions, increasing robustness and eliminating potentially incorrect \textit{a priori} assumptions.} This feature increases robustness and eliminates the need for \textit{a priori} sign assumptions that may be incorrect or difficult to justify in empirical applications.

\paragraph{Extension to heteroscedastic errors.}
Real-world data often exhibit heteroscedasticity in first-stage errors, particularly in cross-sectional and panel settings. We extend the DT condition to accommodate this realistic feature while preserving the core identification logic.

The robust DT condition exploits a key insight:
when the candidate instrument deviates from the true instrument, heteroscedasticity
arises from \emph{two} sources: (1) intrinsic heteroscedasticity associated with the
true instrument, and (2) additional variation induced by instrument misalignment.
When the candidate instrument equals the true instrument, only the intrinsic
heteroscedasticity remains, minimizing the discrepancy between OLS and FGLS variances.

We formalize this by defining a distance criterion $\widehat{D}(\delta)$ that measures
the squared difference between the conditional variances of OLS and FGLS residuals.
The robust SIV is identified by minimizing this distance:
\[
\delta_0 = \arg\min_{\delta \in (0,\bar{\delta})} \widehat{D}(\delta),
\quad\text{yielding}\quad
\mathbf{s}^\star = \mathbf{x} + k\delta_0\mathbf{r}.
\]
This extension maintains computational tractability while preserving instrument validity under heteroscedasticity. We provide both parametric (assuming a specific variance function) and nonparametric (distribution-free) implementations, giving practitioners flexibility based on their data characteristics.

\bigskip
\noindent\textbf{Advantages of the SIV approach.}
Building on this foundation, the SIV framework offers several advantages over conventional IV methods:
(i)  There is no reliance on external instruments as instruments are generated directly from the observed data, eliminating the often-fruitless search for valid external variables..
(ii) The coplanarity relationship ensures that synthetic instruments are strongly correlated with endogenous regressors, avoiding weak instrument problems.
(iii) The method yields a single optimal instrument for each endogenous variable, thus the approach avoids multiple-instrument bias and overidentification issues.
 (iv) The sign of $\operatorname{cov}(\mathbf{x},\mathbf{u})$ is determined from data rather than imposed \textit{a priori}, reducing specification risk.
 (v) 
The robust DT condition accommodates heteroscedastic disturbances without sacrificing identification power.
Collectively, these properties make the SIV method a powerful
data-driven alternative for addressing endogeneity in applied econometrics.

\bigskip
\noindent\textbf{Contributions to the literature.}

Our contributions to the econometric literature are threefold. 

\textit{First}, we provide a novel geometric characterization of valid instruments in the regression plane, demonstrating that instrument construction need not rely on external variables. This geometric perspective reveals fundamental structure that has been implicit but unexploited in existing approaches.

\textit{Second}, we develop the dual tendency condition---a testable criterion linking unobservable exogeneity to observable moment restrictions. This bridges the gap between theoretical requirements for valid instruments and practical data-driven identification.

\textit{Third}, we demonstrate the method's effectiveness through rigorous simulations and diverse empirical applications spanning labor economics, economic history, and policy evaluation. These applications show that SIV produces reliable causal estimates across varied settings where traditional IV methods face challenges.

\bigskip
\noindent\textbf{Relation to existing approaches.}

Several methodological approaches address endogeneity without external instruments, but each imposes specific restrictions that limit applicability. We briefly position SIV relative to these alternatives.

\emph{Moment-based approaches.}
Some methods exploit higher-order moments \citep{lewbel97, erickson2002two} or
heteroscedastic covariance restrictions \citep{Rigobon2003, lewbel, KLEIN2010154}.
These approaches require specific distributional properties (e.g., non-normality)
or covariance structures (e.g., conditional heteroscedasticity) that may not hold
in all applications. The SIV method imposes no such restrictions beyond standard linear model assumptions.

\emph{Latent variable approaches.}
The Latent Instrumental Variable (LIV) framework \citep{ebbes2005solving} decomposes
endogenous regressors into exogenous and endogenous components using latent
discrete variables estimated by maximum likelihood.
This approach requires parametric assumptions about the latent structure and
typically demands large sample sizes for reliable estimation. The SIV method avoids latent structure modeling entirely.

\emph{Design-based synthetic instruments.}
Recent work constructs synthetic instruments from structural constraints:
eigenvector spatial filters \citep{gallo}, sparsity-driven synthetic exposures
\citep{tang2024}, and synthetic controls in IV-DiD settings \citep{vives2023synthetic, abadie2021using}.
While innovative, these methods require additional structure—spatial relationships,
sparsity patterns, or panel data—beyond the basic regression framework. The SIV method operates within the minimal framework of linear regression with endogeneity.

\emph{Copula-based approaches.}
Some studies pursue instrument-free estimation through Gaussian copula dependencies
\citep{haschka2024endogeneity, parkg}.
However, identification requires non-normality and parametric copula specifications
that may be difficult to justify \citep{hashchka}.  The SIV method requires no distributional assumptions beyond those standard in linear IV regression.

\emph{Bartik instruments.}
Conceptually, the SIV method relates to design-based IV strategies (e.g., shift-share
or Bartik instruments; \citealp{bartik1991, blanchard1992, goldsmith}), which construct
instruments from known functional forms combining exogenous shocks with predetermined
variables.
However, unlike these approaches, the SIV framework requires no external shocks
or auxiliary data—the "shock" is implicit in the geometric structure of the regression plane.

In contrast to these approaches, the SIV method imposes no restrictions beyond standard linear model assumptions. 
 The method works within the minimal framework of linear regression with endogeneity, using only the observed relationship between $\mathbf{x}$ and $\mathbf{y}$.

\subsection{Outline of the paper}

The remainder of the paper proceeds as follows.
Section~2 develops the geometric foundation, establishing the properties of coplanar
instruments and deriving the representation $\mathbf{z}_0 = \mathbf{x} + k\delta_0\mathbf{r}$.
Section~3 introduces the dual tendency condition for the homoscedastic case and
presents the identification theorem.
Subsection~\ref{s3.5} extends the method to heteroscedastic errors.
Section~4 demonstrates the method's performance through Monte Carlo simulations
and applies it to three empirical examples: labor supply (Mroz 1987), literacy
and religion (Becker and Woessmann 2009), and retirement savings (Abadie 2003).
Section~5 concludes and discusses extensions to 
nonlinear models, and panel data.

\section{The properties of coplanar IVs}
%
This section develops the geometric foundation of the SIV method. We begin with the intuitive observation that in any regression with endogeneity, the outcome, endogenous regressor, and error term lie in a common two-dimensional plane. This \textit{coplanarity property} implies that valid instruments need not come from outside this plane---they can be constructed from vectors within it. We formalize this insight and derive a parametric representation for all valid instruments, establishing the foundation for our identification strategy in Section \ref{sec:method}.

\subsection{Setup and Notation}

Consider the classical endogeneity problem in a simple structural model:
\begin{eqnarray}\label{e1}
\mathbf {y}= \beta {\mathbf x}+\mathbf u, \\\label{e12}
\mathbf x={\gamma }\mathbf z +\mathbf e
\end{eqnarray}
where $\mathbf y \in \mathbb{R}^n$ is the outcome variable, $\mathbf x \in \mathbb{R}^n$ is the endogenous regressor, $\mathbf u \in \mathbb{R}^n$ is the structural error satisfying $\mathbb{E}(\mathbf u \mid \mathbf x) \neq 0$, and $\mathbf z \in \mathbb{R}^n$ is an instrumental variable.

Throughout, we work in a separable Hilbert space $\mathcal{H}$, where the inner product of two vectors $\mathbf v$ and $\mathbf w$ is denoted $\langle \mathbf v \cdot \mathbf w\rangle$.\footnote{A separable Hilbert space is a complete vector space with an inner product that induces a distance metric and has a countable dense subset. This provides a rigorous mathematical framework for our geometric arguments while generalizing finite-dimensional Euclidean space. For readers unfamiliar with this abstraction, it suffices to think of $\mathbb{R}^n$ with the standard dot product.} All vectors in \eqref{e1}--\eqref{e12} are $n \times 1$ and belong to $\mathcal{H}$.

\paragraph{Handling exogenous controls.} When additional exogenous or predetermined regressors are present, we can reduce the model to the form \eqref{e1}--\eqref{e12} through orthogonal projection. Specifically, let $\mathbf V \in \mathcal{H}^{n \times k}$ denote a matrix of exogenous regressors (including a constant term), and let $\mathbf P_V$ be the orthogonal projection matrix onto $\text{span}(\mathbf V)$. Define the residual vectors:
\[
\mathbf y = (\mathbf I - \mathbf P_V)\tilde{\mathbf y}, \quad \mathbf x = (\mathbf I - \mathbf P_V)\tilde{\mathbf x}, \quad \mathbf z = (\mathbf I - \mathbf P_V)\tilde{\mathbf z},
\]
where $\tilde{\mathbf y}$, $\tilde{\mathbf x}$, and $\tilde{\mathbf z}$ denote the original (unprojected) variables and $\mathbf I$ is the identity matrix. By construction, $\mathbf y$, $\mathbf x$, and $\mathbf z$ are orthogonal to $\text{span}(\mathbf V)$, effectively partialling out the exogenous regressors. We proceed by working with these residual vectors, noting that all results apply to the general case with controls.

\subsection{Standard IV Conditions}

A valid instrumental variable $\mathbf z$ must satisfy two key conditions:

\begin{enumerate}
\item \textbf{Exogeneity}: The instrument is asymptotically uncorrelated with the structural error,
\(
\mathbb{E}(\mathbf u \mid \mathbf z) = 0.
\)
This ensures that $\mathbf z$ is not correlated with unobserved factors influencing the outcome.

\item \textbf{Relevance}: The instrument is sufficiently strongly correlated with the endogenous regressor,
\(
\text{cov}(\mathbf x,\mathbf  z) \neq 0, \quad \text{or equivalently,} \quad \mathbb{E}(\mathbf x \mid \mathbf z) \neq 0.
\)
This ensures that $\mathbf z$ can effectively predict variation in $\mathbf x$. As a rule of thumb, an instrument is considered weak if the first-stage $F$-statistic in \eqref{e12} is less than 10.
\end{enumerate}
Our contribution is to show that these conditions can be satisfied by instruments constructed entirely from the observable vectors $\mathbf x$ and $\mathbf y$, without requiring external data.
 \begin{figure}[ht!]
\centering
\includegraphics[width=0.6\linewidth]{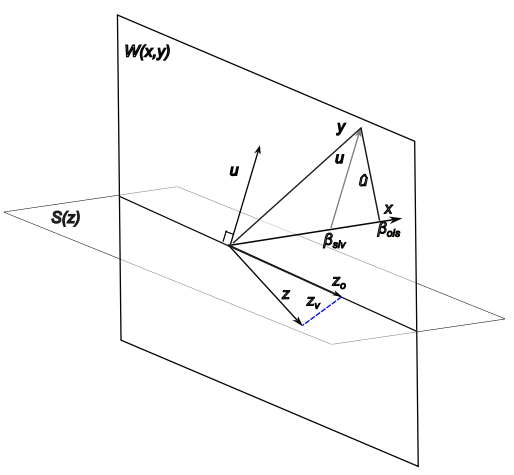}
\caption{\small {\it Geometric representation of instrumental variable (IV) and regression planes.} The figure illustrates the relationship between the outcome variable $\mathbf y$, endogenous regressor $\mathbf x$, error term $\mathbf u$, and instrumental variable $\mathbf z$ in a three-dimensional space. $\mathcal{W}({\mathbf x,\mathbf y})$ represents the plane spanned by $\mathbf x$  and $\mathbf y$, while $\mathcal{S}({\mathbf z})$ is orthogonal to $\mathbf u$.  Vector $\mathbf z_0$ is the component of $\mathbf z$ that is coplanar with $\mathbf x$ and $\mathbf y$ .}
\label{fig2}
\end{figure}

\subsection{The Coplanarity Property}
The key geometric insight underlying our approach is straightforward: the vectors $\mathbf y$, $\mathbf x$, and $\mathbf u$ in equation \eqref{e12} all lie in the same two-dimensional subspace.
\begin{definition}[Coplanarity]
Vectors are \textit{coplanar} if no more than two linearly independent vectors exist in the set. Equivalently, all vectors in the set lie in a common two-dimensional subspace (e.g., in some plane through the origin in $R^n$).
\end{definition}
From equation \eqref{e1}, we have $\mathbf u =\mathbf  y - \beta \mathbf x$, which immediately implies that $\mathbf u$ is a linear combination of $\mathbf y$ and $\mathbf x$. Therefore, $\mathbf y$, $\mathbf x$, and $\mathbf u$ are coplanar by construction---they all lie in the closed linear span of $\mathbf y$ and $\mathbf x$:
 \begin{equation}\label{ne1}\mathcal{W} = \operatorname{span}{( \mathbf x, \mathbf y)} = \left\{\alpha {\mathbf x} + \beta{\mathbf y }\mid \alpha, \beta \in \mathbb{R}.\right\}\end{equation} 

Figure \ref{fig2} illustrates this geometric structure.\footnote{ See  \citet[pp.54-56]{davidson} for a discussion of the geometry of OLS, and a geometric explanation of the IV estimation in \citet{butler}.} The plane $\mathcal{W}$ represents all possible linear combinations of the observable vectors $\mathbf x$ and $\mathbf y$. The structural error $\mathbf u$ must lie within this plane, even though it is unobservable.
\subsection{Projecting Instruments onto the Observable Plane}

While the structural error $\mathbf u$ must lie in $\mathcal{W}$, a general instrumental variable $\mathbf z$ need not. Traditional instruments often come from outside the span of $\mathbf x$ and $\mathbf y$, reflecting external variation or policy shocks. However, we show that for identification purposes, only the component of $\mathbf z$ lying within $\mathcal{W}$ matters.

Let $\mathcal{S}$ denote the space of all vectors satisfying $\mathbb{E}(\mathbf u \mid \mathbf z) = 0$. Any vector in the orthogonal complement of $\mathcal{W}$,
\[
\mathcal{W}^\perp = \{\mathbf z \in \mathcal{H} \mid \langle \mathbf z \mid \mathbf w \rangle = 0 \text{ for all } \mathbf w \in \mathcal{W}\},
\]
automatically satisfies the exogeneity condition, since $\mathbf u \in \mathcal{W}$ implies $\mathbb{E}(\mathbf u \mid \mathbf z) = 0$ for all $\mathbf z \in \mathcal{W}^\perp \equiv \mathcal{S}$.

The following lemma establishes that we can restrict attention to the projection of instruments onto $\mathcal{W}$ without loss of generality.

\begin{lemma}[Exogeneity Preserved Under Projection]\label{r1}
Let $\mathcal W:=\mathcal W(\mathbf x,\mathbf y)=\operatorname{span}\{\mathbf x,\mathbf y\}$ and
let $\mathbf z_0:=P_{\mathcal W}\mathbf z$ be the orthogonal projection of a random vector
$\mathbf z\in \mathcal{H}$ onto $\mathcal W$. Suppose $\mathbf u\in L^1$ and
$\mathbf u\in\mathcal W$.
If $\mathbb E(\mathbf u\mid \mathbf z)=\mathbf 0$ a.s., then
$\mathbb E(\mathbf u\mid \mathbf z_0)=\mathbf 0$ a.s.
\end{lemma}

\begin{proof} See Appendix \ref{ar1}.
\end{proof}
\textbf{Intuition.} Lemma \ref{r1} says that if $\mathbf z$ is a valid instrument, then its projection $\mathbf z_0$ onto the observable plane $\mathcal{W}$ is also a valid instrument. This is because exogeneity depends only on the component of $\mathbf z$ that has the potential to correlate with vectors in $\mathcal{W}$---namely, the component $\mathbf z_0$ within $\mathcal{W}$ itself. The component of $\mathbf z$ orthogonal to $\mathcal{W}$ is by definition uncorrelated with everything in the plane and contributes nothing to identification.

In light of Lemma \ref{r1}, we re-formulate the model \eqref{e1}-\eqref{e12} as follows:
\begin{equation}\label{e2}
{\mathbf y}= \beta {\mathbf x}+\mathbf u, 
\end{equation}
\begin{equation}\label{e2a}
{\mathbf x}={\gamma_0 {\mathbf z}_0} +\mathbf e_0,
\end{equation}
where $\mathbf z_0 \in \mathcal{W}(\mathbf x,\mathbf y) \cap \mathcal{S}(\mathbf z)$ is the coplanar instrument satisfying $\mathbb{E}[\mathbf z_0 \mid \mathbf u] = 0$.

\textbf{Key implication.} Since all vectors required for IV estimation can be contained within the observable plane $\mathcal{W}(\mathbf x,\mathbf y)$, we limit our focus to this plane when synthesizing instruments. This dramatically simplifies the search problem: instead of seeking an instrument in the high-dimensional space $\mathcal{H}$, we need only search within the two-dimensional plane $\mathcal{W}$.

\subsection{Parametric Representation of Coplanar Instruments}
Having established that valid instruments can be found within $\mathcal{W}$, we now derive their parametric structure. The next lemma shows that any coplanar instrument can be written as a linear combination of two basis vectors: the endogenous regressor $\mathbf x$ and a direction vector $\mathbf r$ orthogonal to $\mathbf x$.

\begin{lemma}[Linear Representation]\label{t2}
Let $\mathcal W(\mathbf x,\mathbf y)=\operatorname{span}\{\mathbf x,\mathbf y\}$ and suppose 
$\mathbf r\in \mathcal W(\mathbf x,\mathbf y)$ is linearly independent of $\mathbf x$. 
Let $\mathbf z_0$ be a valid instrument satisfying 
$\mathbf z_0\in \mathcal W(\mathbf x,\mathbf y)$ and $\,E(\mathbf u\mid \mathbf z_0)=0$.
Then $\mathbf z_0$ can be written as a linear combination of $\mathbf x$ and $\mathbf r$:
\[
\mathbf z_0 \;=\; \zeta\,\mathbf x \;+\; \omega\,\mathbf r
\qquad\text{for some } \zeta,\omega\in\mathbb R.
\]
\end{lemma}

\begin{proof}
By assumption $\mathbf x,\mathbf r\in\mathcal W(\mathbf x,\mathbf y)$ and 
$\mathbf r$ is linearly independent of $\mathbf x$. Hence $\{\mathbf x,\mathbf r\}$ is a basis of the
two--dimensional subspace $\mathcal W(\mathbf x,\mathbf y)$. Since $\mathbf z_0\in\mathcal W(\mathbf x,\mathbf y)$,
the standard linear algebra result on bases implies that there exist unique scalars
$\zeta,\omega\in\mathbb R$ such that
\(
\mathbf z_0 \;=\; \zeta\,\mathbf x \;+\; \omega\,\mathbf r.
\)$\square$
\end{proof}
Lemma \ref{t2} provides the foundation for our parametric approach, but it does not yet pin down a unique representation. To obtain a useful parametrization, we impose convenient geometric restrictions on the direction vector $\mathbf r$ and the coefficient scaling.

\subsection{Normalized Representation with Sign Information}

We now impose additional structure to obtain a parsimonious and economically interpretable representation. The key is to choose the direction vector $\mathbf r$ to be orthogonal to $\mathbf x$, and to normalize the coefficient on $\mathbf x$ to unity.
\textbf{Constructing the direction vector $\mathbf r$.} Define $\mathbf r$ as the residual from projecting $\mathbf y$ onto $\mathbf x$:
\[
\mathbf r = (\mathbf I - \mathbf P_x)\mathbf y, \quad \text{where} \quad \mathbf P_x := \mathbf x(\mathbf x'\mathbf x)^{-1}\mathbf x'
\]
is the orthogonal projection matrix onto the span of $\mathbf x$. By construction:
\begin{itemize}
\item $\mathbf r \perp\mathbf  x$ (i.e., $\text{corr}(\mathbf x,\mathbf r) = 0$),
\item $\mathbf r \in \mathcal{W}(\mathbf x,\mathbf y)$ (since $\mathbf r$ is a linear combination of $\mathbf x$ and $\mathbf y$),
\item $\mathbf r$ captures the component of $\mathbf y$ orthogonal to $\mathbf x$.
\end{itemize}

\textbf{Correlation restrictions.} We impose the following restrictions on the correlations among $\mathbf x$, $\mathbf r$, and $\mathbf z_0$:
\begin{align}
&\operatorname{corr}(\mathbf x, \mathbf r) = 0, \label{eq:corr1} \\
&\operatorname{corr}(\mathbf x,\mathbf  z_0) > 0, \label{eq:corr2}\\
&\operatorname{corr}(\mathbf r, \mathbf u) > 0, \label{eq:corr3} \\
&\operatorname{corr}(\mathbf y, \mathbf r) > 0. \label{eq:corr4}
\end{align}

These restrictions pin down a convenient geometry in $\mathcal{W}(\mathbf x,\mathbf y)$ for the subsequent construction and identification of valid synthetic instruments (see Figure \ref{fig4}). Condition \eqref{eq:corr1} is satisfied by construction. Condition \eqref{eq:corr2} restricts attention to instruments positively correlated with the endogenous variable (the sign can be reversed by multiplying through by $-1$ if needed). 
We assume $\operatorname{corr}(\mathbf y,\mathbf r)>0$, and since
$\mathbf y=\beta\mathbf x+\mathbf u$ with $\operatorname{cov}(\mathbf x,\mathbf r)=0$, which implies condition \ref{eq:corr3}:   $\operatorname{cov}(\mathbf u,\mathbf r)>0$.
Condition \eqref{eq:corr4} ensures that both $\mathbf r$ and $\mathbf y$ point in the same general direction within $\mathcal{W}$.
With these restrictions in place, we can now state the main geometric result.

\begin{lemma}[Normalized Representation with Endogeneity Sign]\label{lb2}
Let $\mathcal W(\mathbf x,\mathbf y)=\operatorname{span}\{\mathbf x,\mathbf y\}$ be the plane
spanned by the endogenous regressor $\mathbf x$ and the outcome $\mathbf y$. 
Suppose there exist noncollinear vectors $\mathbf x,\mathbf r,\mathbf z_0 \in \mathcal W(\mathbf x,\mathbf y)$ satisfying conditions \eqref{eq:corr1}--\eqref{eq:corr4}

Then the vector $\mathbf z_0\in\mathcal W(\mathbf x,\mathbf y)$ can be written as a linear combination of
$\mathbf x$ and $\mathbf r$ of the form
\begin{equation}\label{siv}
\mathbf z_0 \;=\; \mathbf x \;+\; k\,\delta\,\mathbf r,
\end{equation}
for some $\delta\in\mathbb R$, where $\mathbf r$ satisfies $\mathbb{E}(\mathbf r'\mathbf x)=0$ and
\[
k \;:=\; -\,\operatorname{sign}\!\big(\operatorname{cov}(\mathbf x,\mathbf u)\big)\in \{-1, +1\}
\]
\end{lemma}

\begin{proof}See Appendix \ref{alb2}.\end{proof}
\textbf{Interpretation.} Lemma \ref{lb2} establishes the key parametric representation underlying the SIV method. Several features deserve emphasis:

\begin{enumerate}
\item [i.]\emph{One-dimensional search.} The entire family of potential coplanar instruments is parameterized by a single scalar $\delta \geq 0$. Instead of searching for an instrument in an infinite-dimensional space, we need only search along a one-dimensional path.

\item[ii.] \emph{Sign encodes endogeneity direction.} The parameter $k \in \{-1, +1\}$ captures the direction of endogeneity. When $\text{cov}(\mathbf x,\mathbf u) > 0$, we have $k = -1$; when $\text{cov}(\mathbf x,\mathbf u) < 0$, we have $k = +1$. This sign will be determined empirically from the data (Subsection \ref{sign}).

\item[iii.] \emph{Geometric interpretation.} The representation $\mathbf z_0 = \mathbf x + k\delta_0\mathbf  r$ says that the valid instrument is obtained by starting at $\mathbf x$ and moving a distance $\delta_0$ in the direction of $\pm \mathbf r$ (the sign depending on the direction of endogeneity). Figure \ref{fig4} illustrates this geometry for both cases.

\item[iv.] \emph{Normalization.} By writing $\mathbf z_0 = \mathbf x + k\delta_0 \mathbf r$ rather than $\mathbf z_0 = \zeta\mathbf x + \omega \mathbf r$, we normalize the coefficient on $\mathbf x$ to unity. This eliminates scale indeterminacy and makes $\delta_0$ interpretable as a "distance" from $\mathbf x$ in the $\mathbf r$ direction.
\end{enumerate}

\begin{figure}[!hb]
	\centering\begin{tabular}{ l   l  }
		\begin{minipage}{.45\textwidth}
			\includegraphics[scale=0.8]{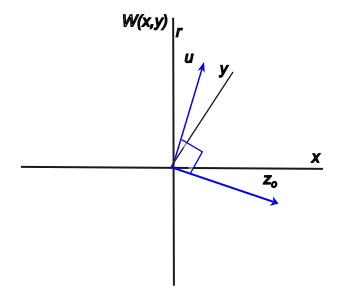}
		\end{minipage}
		&\begin{minipage}{.5\textwidth}
			\includegraphics[scale=0.8]{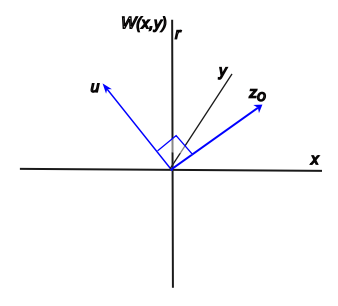}
		\end{minipage}\\
{\text{a. $cor({\mathbf x,u})>0$}}&{\text{b. $cor({\mathbf x,u})<0$}}\\
\end{tabular}\caption{\small {\it Orientation of   a valid SIV ${\mathbf z_0}\in \mathcal{W}({\mathbf x,y})$ relative to $\mathbf y$, $\mathbf x$ given the error term $\mathbf u$.}
Panel (a) shows the case when $cor({\mathbf x,u}) > 0$, and panel (b) shows the case when $cor({\mathbf x,u}) < 0$. %
}\label{fig4}\end{figure}

\label{pp}
\subsection{From Geometry to Identification}

The geometric results in this section reduce the instrument search problem to finding two unknowns:

\begin{enumerate}
\item The \textbf{sign} $k \in \{-1, +1\}$, which encodes the direction of endogeneity.
\item The \textbf{scale} $\delta_0 > 0$, which determines the location of the valid instrument along the search path.
\end{enumerate}

Once these are determined, the synthetic instrument is:
\[
\mathbf s^\star = \mathbf x + k\delta_0\mathbf  r.
\]

The challenge is that both $k$ and $\delta_0$ depend on the unobservable error $\mathbf u$ through the exogeneity condition $\mathbb{E}(\mathbf u \mid \mathbf z_0) = 0$. Section \ref{sec:method} develops the \textit{dual tendency} (DT) condition, which provides testable moment restrictions that identify both $k$ and $\delta_0$ from observable data.

\section{Synthetic Instrumental Variable (SIV) Method}\label{sec:method}

This section develops a \emph{synthetic instrumental variable} (SIV)
$\mathbf{s}^\star\in\mathcal{W}(\mathbf{x},\mathbf{y})\subseteq\mathcal{H}$
satisfying \mbox{$\mathbb{E}(\mathbf{u}\mid\mathbf{s}^\star)=0$.}
The SIV is constructed from the observed pair $(\mathbf{x},\mathbf{y})$ and therefore does
not rely on external instruments.

The construction is parameterized by a scalar $\delta>0$ and a direction
$\mathbf{r}$ that is orthogonal to $\mathbf{x}$ but coplanar with
$(\mathbf{x},\mathbf{y})$.
Identification proceeds by imposing a \emph{dual-tendency} (DT) condition: a set of moment
conditions implied by homoscedasticity, or by a transformed system under
heteroscedasticity.
We first state the assumptions, then derive the DT condition and identification in the
homoscedastic case, and finally introduce a robust version for heteroscedastic errors.

\subsection{Assumptions for the SIV Framework}

\begin{definition}[A1. Coplanarity]
The vectors $\mathbf x$, $\mathbf r$, and $\mathbf s$ belong to
$\mathcal W(\mathbf x,\mathbf y)\subseteq\mathcal H$ and are thus coplanar.
By Lemma~\ref{lb2}, an SIV equals
\[
\mathbf s \;=\; \mathbf x + k\,\delta\,\mathbf r, 
\qquad
\delta > 0,\quad
k := -\,\operatorname{sign}\!\big(\operatorname{cov}(\mathbf x,\mathbf u)\big).
\]
The auxiliary vector $\mathbf r$ satisfies $\mathbb{E}[\mathbf r'\mathbf x]=0$ and is defined as
\[
\mathbf r \;=\; (\mathbf I - \mathbf P_x)\,\mathbf y,
\qquad
\mathbf P_x := \mathbf x(\mathbf x'\mathbf x)^{-1}\mathbf x'
\]
the orthogonal projection matrix onto the span of $\mathbf x$.
The data $\{\mathbf Z_i\}_{i=1}^n$, where $\mathbf Z_i$ collects the $i$th observations on
$(\mathbf y,\mathbf x,\mathbf r)$, are assumed i.i.d.\ (or stationary and ergodic).
\end{definition}

\begin{definition}[A2. Synthetic Instrument]\label{asmp:siv_construction}
The SIV $\mathbf s$ is an $n\times 1$ vector defined by
\[
\mathbf s \;=\; \mathbf x + k\,\delta\,\mathbf r,
\]
with $\delta\in(0,\bar\delta)$, where $\bar\delta$ is a finite upper bound chosen to rule out arbitrarily weak instruments
(e.g. by requiring $\operatorname{corr}(\mathbf s,\mathbf x)\ge c>0$).
The SIV satisfies the exclusion restriction: $\mathbf s$ affects $\mathbf y$ only through
its $\mathbf x$-component, not through $\mathbf r$ \citep{HECKMAN2024105719}.

\end{definition}

\begin{definition}[A3. Full Rank]
The expectations $\mathbb{E}[\mathbf x'\mathbf y]$ and $\mathbb{E}[\mathbf x'\mathbf x]$ exist and are identified
from the data. $\mathbb{E}\big[\mathbf x\,\mathbf x'\big]$ exists and is positive definite, which reduces to $\mathbb{E}\big[x^2\big]>0$ in the univariate case

\end{definition}

\begin{definition}[A4. First-Stage Error Structure]
The first-stage error term $\mathbf e=(e_1,\dots,e_n)'\in \mathbb R^{n}$ is assumed independent across $i$
(and identically distributed when indicated) and may exhibit either homoscedasticity or
heteroscedasticity.

\textbf{Homoscedastic case.} In the single-equation ($p=1$) and single-IV model ($q=1$),  
the conditional second moment for a valid IV $\mathbf z_0$, $\mathbb{E}(\mathbf e\mathbf e' \mid\mathbf z_0) =\mathbf H = \sigma^2 \mathbf I_n$ satisfies
\[
 \mathbb E\!\big[\mathbf e^{\circ 2}\mid \mathbf z_0\big] 
= \sigma^2\,\mathbf 1_n,
\quad\text{i.e.}\quad
 \mathbb E(e_i^2\mid \mathbf z_0)=\sigma^2\quad\forall i,
\]
where $\mathbf e^{\circ 2}$ is the Hadamard (entrywise) square
and $\mathbf 1_n$ is the $n$-vector of ones.

\textbf{Heteroscedastic case.} We allow the conditional variance to depend on observed covariates. For the single-equation model, write
\[
\mathbb E(e_i^2\mid \mathbf Z_i) 
\;=\; H_i(\boldsymbol\theta)
\;=\; h\!\big(\hat\sigma^2 + \mathbf Z_i'\boldsymbol\theta\big),
\]
where $h:\mathbb R\to(0,\infty)$ is twice continuously differentiable and applied elementwise,
$\mathbf Z_i\in\mathbb R^{k}$ contains the $k$ exogenous or predetermined variables
relevant for the variance, and $\boldsymbol\theta\in\mathbb R^{k}$ is a parameter vector.
A feasible estimate replaces ${e}_i^2$ by $\hat e_i^{\,2}$, yielding
\[
\mathbb E(\hat e_i^{\,2}\mid \mathbf Z_i)
\approx H_i(\boldsymbol\theta)
= h\!\big(\hat\sigma^2 + \mathbf Z_i'\boldsymbol\theta\big),
\]
which provides a consistent estimate of the conditional variance function
under standard regularity conditions.
\end{definition}
These assumptions are standard in the IV literature, with the exception that Assumption A2 specifies that instruments are constructed synthetically rather than taken as given external variables.
\subsection{Dual Tendency (DT) Condition for Coplanar IVs}
The dual tendency condition exploits a fundamental link between two properties that hold simultaneously for a valid instrument: exogeneity and first-stage homoscedasticity. We begin by characterizing the first-stage coefficient structure.

\subsubsection{First-Stage Coefficient Decomposition}

When we use a candidate SIV $\mathbf s = \mathbf x + k\delta \mathbf r$ to instrument for $\mathbf x$ in equation \eqref{e1}, the first-stage regression coefficient depends on how closely $\mathbf s$ approximates the true instrument $\mathbf z_0 = \mathbf x + k\delta_0 \mathbf r$.

\begin{lemma}\label{la1}
Under Assumptions A1--A3, let $\mathbf z_0\in\mathcal W(\mathbf x,\mathbf y)$ be a valid IV satisfying
$\mathbb E(\mathbf u\mid \mathbf z_0)=0$, and let the first-stage equation be
\[
\mathbf x \;=\; \gamma\,\mathbf s + \mathbf e,
\]
where $\mathbf s$ is a synthetic instrument (SIV) of the form
$\mathbf s = \mathbf x + k\,\delta\,\mathbf r$. Then the corresponding first-stage
coefficient can be written as
\[
\gamma \;=\; \gamma_0 + g(\mathbf x,\mathbf s,\mathbf z_0),
\]
where
\[
\gamma_0 \;:=\; \frac{\operatorname{cov}(\mathbf x,\mathbf z_0)}{\operatorname{var}(\mathbf z_0)}
\]
is the population OLS coefficient when instrumenting with the true IV $\mathbf z_0$, and
$g(\cdot)$ is a (locally) twice continuously differentiable function capturing the deviation
of $\mathbf s$ from $\mathbf z_0$.
\end{lemma}

\begin{proof}
See Appendix \ref{ala1}.
\end{proof}
\textbf{Intuition.} When $\mathbf s = \mathbf z_0$, we have $g(\mathbf x,\mathbf s,\mathbf z_0) = 0$ and $\gamma = \gamma_0$. As $\mathbf s$ deviates from $\mathbf z_0$ (i.e., as $\delta$ moves away from $\delta_0$), the first-stage coefficient changes smoothly. This smoothness will be exploited to establish continuity of the moment function below.
\subsubsection{The Moment Restriction}

Homoscedasticity yields a set of moments characterizing the valid IV.
 Let $\mathbf e = \mathbf x - \gamma \mathbf s$ denote the first-stage residual. When  $p=1$ and $\mathbf s = \mathbf z_0$ is the true instrument, Assumption A4 implies
\[
\mathbb{E}\!\big[\mathbf{e}^{\circ 2}\mid\mathbf{z}_0\big]
=
\sigma^2\,\mathbf{1}_n.
\]
In the single-equation, single-IV case ($p=q=1$), the $i$th moment condition is
\[
m_i(\delta)
=
( e_i^2-\sigma^2 )\,z_{0i},
\]
and the stacked moment vector is
\begin{equation}\label{m:1}
\mathbf{M}(\delta)
=
\mathbb{E}\big[(\mathbf{e}(\delta)^{\circ 2}-\sigma^2\mathbf{1}_n)\, \odot \mathbf{z}_0\big]\;=\;
\mathbb E\big[(\mathbf e(\delta)^2-\sigma^2)\,\mathbf z_0\big],
\end{equation}
where $\odot$ denotes the Hadamard (elementwise) product.

\begin{lemma}[Dual Tendency Condition for Valid Instruments]\label{t0a}
Under Assumptions A1--A3 and the homoscedasticity condition
\[\mathbb
E\!\big(\mathbf e^{\circ 2}\mid \mathbf z_0\big)
= \sigma^2\,\mathbf 1_n,
\]
the orthogonality condition
\[\mathbb
 E(\mathbf u\mid \mathbf z_0)=0
\]
and the  moment condition ("first-stage homoscedasticity")
\[
\mathbf M(\delta_0)=\mathbf 0
\]
hold simultaneously for any valid IV $\mathbf z_0\in \mathcal{W}(\mathbf{x},\mathbf{y})$ defined as $\mathbf z_0 \;=\; \mathbf x \;+\; k\,\delta_0\,\mathbf r$.
\end{lemma}

\begin{proof}
See Appendix~\ref{at0a}.
\end{proof}

We refer to this as the \emph{dual tendency} (DT) condition,
which implies that a valid instrument must satisfy both
$\mathbb{E}(\mathbf{u}\mid\mathbf{z}_0)=0$
and $\mathbf{M}(\delta)=\mathbf 0$.
The condition applies to $\mathbf{z}_0$, the coplanar projection of an underlying
(possibly external) instrument $\mathbf{z}$ onto $\mathcal{W}(\mathbf{x},\mathbf{y})$.

\textbf{Interpretation:} Why "dual tendency"? The name reflects the fact that two distinct conditions---one unobservable (exogeneity) and one testable (the moment restriction)---hold together if and only if we have chosen the correct instrument. Think of it as two independent tests that both point to the same answer. When we search over $\delta$, there is generically only one value where both tests align: $\delta = \delta_0$.

This duality provides the foundation for identification: by finding the value of $\delta$ that satisfies the observable condition $\mathbf M(\delta) = \mathbf 0$, we simultaneously identify the instrument that satisfies the unobservable condition $\mathbb{E}(\mathbf u \mid\mathbf z_0) = 0$.

\subsection{Identification via the SIV Method: Homoscedastic Case}\label{identification}

Any $\mathbf{s}\in\mathcal{W}(\mathbf{x},\mathbf{y})$ can be written uniquely as
\[
\mathbf{s}(\delta)
=
\mathbf{x} + k\,\delta\,\mathbf{r},
\qquad
\delta\in(0,\bar{\delta}).
\]
The target is to recover $\delta_0$ such that $\mathbf s^\star:=\mathbf{s}(\delta_0)=\mathbf{z}_0$, i.e.\ the
value for which the DT condition holds.

 \begin{theorem}[Identification of SIV via the DT Condition]\label{lm10}
Suppose Assumptions A1--A3 and the homoscedastic part of Assumption A4 hold. 
Let there exist a valid coplanar instrument 
$\mathbf z_0 \in \mathcal W(\mathbf x,\mathbf y)$ such that
\[
\mathbb E(\mathbf u \mid \mathbf z_0) = \mathbf 0
\quad\text{and}\quad
\mathbb E(\mathbf e \mathbf e' \mid \mathbf z_0) = \sigma^2 \mathbf I_n,
\]
where $\mathbf e$ is the first–stage residual.

Let $\mathbf s(\delta) = \mathbf x + k\,\delta\,\mathbf r$, $\delta \in (0,\bar\delta)$,
parameterize the class of candidate synthetic instruments as in Lemma~\ref{lb2}, and let
$\delta_0 \in (0,\bar\delta)$ be such that $\mathbf s(\delta_0) = \mathbf z_0$.
Let $\mathbf M(\delta)$ denote the DT moment vector defined in \eqref{m:1}.

Then $\mathbf M(\delta_0) = \mathbf 0$. If $\mathbf M(\delta)$ is continuously differentiable
in $\delta$ on $(0,\bar\delta)$ and its derivative at $\delta_0$,
\[
J_{\mathbf M}(\delta_0)
:= \frac{\partial \mathbf M(\delta)}{\partial \delta}\bigg|_{\delta=\delta_0},
\]
is nonzero, then there exists a neighborhood $\mathcal N$ of $\delta_0$ such that
\[
\mathbf M(\delta) = \mathbf 0 
\quad\Longleftrightarrow\quad 
\delta = \delta_0,
\qquad (\delta \in \mathcal N).
\]
Hence $\delta_0$ is locally identified by the DT condition, and the synthetic instrument
\[
\mathbf s^\star := \mathbf s(\delta_0) = \mathbf z_0
\]
is the unique coplanar instrument in $\mathcal W(\mathbf x,\mathbf y)$ satisfying
$\mathbb E(\mathbf u \mid \mathbf s^\star) = \mathbf 0$.
\end{theorem}

\begin{proof}
See Appendix~\ref{alm10}.\end{proof}
Theorem~\ref{lm10} identifies $\delta_0$ as the unique solution to
$\mathbf{M}(\delta)=\mathbf{0}$ in a neighborhood of the true SIV.
In practice, $\mathbf{M}(\delta)$ is replaced by its sample analogue, and $\delta_0$ is
estimated by minimizing a quadratic form in the sample moments.

\subsubsection{Consistency of the DT Estimator}\label{consist:dt}
Having established identification, we now show that the sample estimator converges to the true parameter value.
\begin{corollary}[DT Estimator]\label{comp}
Under the conditions of Theorem~\ref{lm10}, the parameter $\delta_0$ can be consistently
estimated by
\[
\hat{\delta}_n
=
\arg\min_{\delta\in\mathcal{D}} \widehat{J}_n(\delta),
\]
where $\mathcal{D}\subset(0,\bar{\delta})$ is compact and
\[
\widehat{J}_n(\delta)
=
\widehat{\mathbf{M}}_n(\delta)'\mathbf{W}_n\widehat{\mathbf{M}}_n(\delta)
\]
is a sample criterion formed from the sample moment vector
$\widehat{\mathbf{M}}_n(\delta)$ and a positive definite weighting matrix $\mathbf{W}_n$.
Under a uniform law of large numbers and $\mathbf{W}_n\xrightarrow{p}\mathbf{W}$,
$\hat{\delta}_n\xrightarrow{p}\delta_0$.
\end{corollary}
\begin{proof}
See Appendix \ref{acomp}.
\end{proof}

With the valid SIV $\mathbf s^\star$ in hand, we can estimate the structural parameter $\beta$ using standard instrumental variables techniques.
$\hat{\mathbf{s}} = \mathbf{x} + k\hat{\delta}_n\mathbf{r}$.
\begin{corollary}[Standard SIV Estimator]\label{tx1}
If $\mathbf s^\star=\mathbf x + k\,\delta_0\,\mathbf r$ satisfies
Theorem~\ref{lm10}, then the structural parameter $\beta$ in \eqref{e2} is identified, and its
population analogue can be estimated by
\[
\hat\beta_{\text{SIV}}
\;=\; (\mathbf s^{\star\prime}\mathbf x)^{-1}\,\mathbf s^{\star\prime}\mathbf y.
\]
\end{corollary}

\begin{proof}
See Appendix~\ref{atx1}.
\end{proof}

\subsection{Identifying the Sign of $\operatorname{cov}(\mathbf{x},\mathbf{u})$}\label{sign}

Lemma~\ref{lb2} implies that correct sign specification is required for
identifying $\mathbf{s}^\star$.
From Theorem~\ref{lm10}, $\mathbf{s}^\star$ satisfies $\mathbf{M}(\delta_0)=0$, which implies
$\operatorname{cov}(\mathbf{e}^{\circ 2},\mathbf{s}^\star)=0$
for $\delta_0>0$.
If the sign of $\operatorname{cov}(\mathbf{x},\mathbf{u})$ is misspecified, the DT
condition fails for all $\delta>0$.

\begin{corollary}[Sign Determination]\label{c2a}
Under Assumptions A1--A3 and homoscedastic $\mathbf{u}$ with
$\operatorname{corr}(\mathbf{x},\mathbf{u})\neq 0$,
the true sign of $\operatorname{corr}(\mathbf{x},\mathbf{u})$
is that which yields $\delta_0>0$ and satisfies
$\operatorname{cov}(\mathbf{e}^{\circ  2},\mathbf{s}^\star)=0$,
for $\mathbf{s}^\star = \mathbf{x} + k\delta_0\mathbf{r}$.
\end{corollary}
\begin{proof}
See Appendix~\ref{ac2a}.
\end{proof}

\textbf{Implementation.} In practice, we test both sign assumptions ($k = +1$ and $k = -1$) separately:

\begin{enumerate}
\item For $k = +1$: Search for $\delta_0^{(+)} > 0$ satisfying $\mathbf M(\delta_0^{(+)}) = 0$ with $\mathbf s^{(+)} = \mathbf x + \delta_0^{(+)}\mathbf  r$.
\item For $k = -1$: Search for $\delta_0^{(-)} > 0$ satisfying $\mathbf M(\delta_0^{(-)}) = 0$ with $\mathbf s^{(-)} = \mathbf x - \delta_0^{(-)}\mathbf  r$.
\end{enumerate}

Under the correct sign assumption, the function $\text{cov}(\mathbf e^{\circ 2}, \mathbf s(\delta))$ will exhibit a nonmonotonic pattern, crossing zero at $\delta = \delta_0$. Under the incorrect sign assumption (or if there is no endogeneity), the covariance will not cross zero and typically increases monotonically in magnitude. This provides a diagnostic for identifying the true sign and detecting the absence of endogeneity.

\begin{remark}
If neither sign yields a valid zero-crossing, this suggests $\text{cov}(\mathbf x,\mathbf u) \approx 0$---i.e., no endogeneity is present. In this case, OLS is consistent and no instrumental variable correction is needed.
\end{remark}
Having established identification in the homoscedastic case, we now extend the method to accommodate heteroscedastic errors—a more realistic setting for empirical applications.

\subsection{Extension: The Robust DT Condition for Heteroscedastic Errors}
\label{s3.5}

%
The DT condition in Subsection \ref{identification} relies on first-stage homoscedasticity: $\mathbb{E}(\mathbf e^{\circ 2} \mid \mathbf z_0) = \sigma^2 \mathbf{1}_n$. In practice, first-stage errors often exhibit heteroscedasticity, particularly in cross-sectional data. We now extend the DT framework to accommodate this realistic feature while preserving identification.
\subsubsection{The Challenge of Heteroscedasticity}

When $\mathbb{E}(\mathbf e\mathbf e' \mid\mathbf z_0) =\mathbf H \neq \sigma^2 \mathbf I_n$, the condition $\mathbf M(\delta) = \mathbf 0$ may no longer hold exactly at the true instrument. The problem is that heteroscedasticity introduces additional variation in $\mathbf e^{\circ 2}$ that is not related to instrument misspecification.

\textbf{Key insight.} While we cannot eliminate heteroscedasticity, we can distinguish between two sources of variance in the first-stage residuals:

\begin{enumerate}
\item[i.] \emph{Intrinsic heteroscedasticity}: Variance inherent to the data generating process, present even when using the true instrument $\mathbf z_0$.
\item[ii.] \emph{Misspecification-induced variance}: Additional variance arising when the candidate instrument $\mathbf s(\delta)$ deviates from $\mathbf z_0$.
\end{enumerate}

When $\mathbf s(\delta) = \mathbf z_0$, only intrinsic heteroscedasticity remains. As $\mathbf s(\delta)$ moves away from $\mathbf z_0$, misspecification-induced variance appears. The robust DT condition identifies $\delta_0$ as the point where this additional variance is minimized.

\subsubsection{GLS Transformation and Variance Comparison}
When for the first-stage stacked scalar error terms, $\mathbb{E}(\mathbf e\mathbf e' \mid \mathbf z_0) = \mathbf H \neq \sigma^2 \mathbf I_n$,
the Generalized Least Squares (GLS) transformation restores spherical
disturbance properties:
\begin{equation}\label{s:1}
\mathbf P'\mathbf x \;=\; \mathbf P'\mathbf s\,\gamma \;+\; \mathbf P'\mathbf e,
\qquad
\mathbf P'\mathbf P \;=\; \mathbf H^{-1}.
\end{equation}
This implies the transformed residuals
\begin{equation}\label{s:2}
\mathbf e_g \;=\; \mathbf P'\mathbf e .
\end{equation}

If $\mathbf H$ were known, the stacked scalar GLS residuals would satisfy 
\(\mathbb
 E(\mathbf e_g \mathbf e_g' \mid \mathbf z_0) \;=\; \mathbf I_n
\), and we could apply the standard DT condition to $\mathbf e_g$ rather than $\mathbf e$.
In practice, $\mathbf H$ is unknown and must be estimated, yielding
OLS and feasible GLS (FGLS) residuals
$\hat{\mathbf e}$ and $\hat{\mathbf e}_g$, respectively
\citep[see][Ch.~8]{hill}. Since $\hat{\mathbf e}_g \not\equiv \mathbf e_g$ in finite samples, neither set of residuals perfectly satisfies the spherical property. However, the \textit{discrepancy} between OLS and FGLS variances contains information about instrument validity.
\subsubsection{Population Criterion}
For each candidate instrument $\mathbf s(\delta)$, let:
\begin{itemize}
\item $\mathbf e(\delta)$ denote OLS residuals from regressing $\mathbf x$ on $\mathbf s(\delta)$,
\item $\mathbf e_g(\delta)$ denote FGLS residuals (using estimated variance weights).
\end{itemize}
Given we have scalar ($p=1$) and single instrument( $q=1$) model, we define the variance discrepancy:
\begin{equation}
\Delta(\delta) := \mathbb{E}[\mathbf e_g^{\circ 2}(\delta) \mid\mathbf  s] - \mathbb{E}[\mathbf e^{\circ 2}(\delta) \mid\mathbf  s] = \mathbf{1}_n - \text{diag}(\mathbf H(\delta)),
\label{eq:variance_discrepancy}
\end{equation}
where $\mathbf H(\delta) := \mathbb{E}[\mathbf e(\delta)\mathbf e(\delta)' \mid \mathbf s]$. $\Delta(\delta)$ measures how far the conditional variance of the first-stage residual under the synthetic instrument deviates from the homoscedastic benchmark.
\noindent
%
Note that
\[
\operatorname{diag}\big(\mathbf H(\delta)\big) - \mathbf{1}_n=\operatorname{diag}\big(\mathbf H(\delta) - \mathbf I_n\big)
= -\Delta(\delta),
\]
so the diagonal entries of $\mathbf H(\delta)-\mathbf I_n$ are
\[
\big[\mathbf H(\delta)-\mathbf I_n\big]_{ii}
= H_{ii}(\delta)-1
= -\Delta_i(\delta).
\]
Thus, to quantify the variance discrepancy more generally, we can use the squared Frobenius norm:
\begin{equation}\label{xx1}
D(\delta)
\;=\; \big\|\mathbf H(\delta)-\mathbf I_n\big\|_F^2
\;=\; \operatorname{tr}\!\Big(\big[\mathbf H(\delta)-\mathbf I_n\big]^2\Big),
\end{equation}
where \(\|\cdot\|_F\) denotes the Frobenius norm, which captures both  diagonal and off-diagonal contributions to variance discrepancy. One can show that
\noindent
%
\[
D(\delta)
= \big\|\Delta(\delta)\big\|_2^2
  \;+\; \sum_{i\neq j} H_{ij}(\delta)^2
\;\ge\;
\big\|\Delta(\delta)\big\|_2^2,
\]
with equality if and only if $H_{ij}(\delta)=0$ for all $i\neq j$ (i.e.\ $\mathbf H(\delta)$
is diagonal).


Under heteroscedasticity, the infimum of \(D(\delta)\) is generally strictly positive.
 However, $D(\delta)$ is minimized at the true instrument:
\[
\delta_0 \in \arg\min_{\delta \in (0,\bar{\delta})} D(\delta).
\]
\textbf{Intuition.} When $\mathbf s(\delta) = \mathbf z_0$, the OLS residuals $\mathbf e(\delta)$ reflect only intrinsic heteroscedasticity, and the FGLS procedure (which targets this intrinsic structure) works as well as possible. The transformed residuals $\mathbf e_g(\delta)$ are close to spherical, making $D(\delta)$ small. When $\mathbf s(\delta) \neq \mathbf z_0$, additional misspecification-induced variance enters, increasing the discrepancy between OLS and FGLS variances, and $D(\delta)$ rises.

\subsubsection{Sample Criterion and Estimation}
To implement this criterion, we model the conditional variance structure. Under Assumption~A4, suppose the conditional variance follows the linear form
\[
\sigma_i^2(\delta)
\;=\;  E\!\big(e_i^2 \mid s_i, z_{0i}\big)
\;=\; \sigma_0^2 \;+\; \zeta\, d_i(\delta) \;+\; \alpha\, z_{0i},
\qquad d_i(\delta):=s_i - z_{0i}.
\]
Estimate the sample analogue using
\begin{equation}\label{x1}
\hat e_i^{\,2} \;=\; b \;+\; \zeta\, d_i(\delta) \;+\; \alpha\, z_{0i} \;+\; v_i,
\end{equation}
and define fitted conditional variances
\begin{equation}\label{x1s}
\hat\sigma_i^{\,2}(\delta) \;=\; \hat b \;+\; \hat\zeta\, d_i(\delta) \;+\; \hat\alpha\, z_{0i},
\quad i=1,\ldots,n.
\end{equation}
Assuming no serial correlation across the stack, set
\begin{equation}\label{xs:1}
\widehat{\mathbf H}_n(\delta)
\;:=\; \operatorname{diag}\!\big(\hat\sigma_1^{\,2}(\delta),\ldots,\hat\sigma_n^{\,2}(\delta)\big)
\;\approx\; E\!\big[\mathbf e(\delta)\mathbf e(\delta)'\,\big|\,\mathbf s\big].
\end{equation}
The sample criterion is
\begin{equation}\label{s:4}
\widehat D_n(\delta)
\;=\; \operatorname{tr}\!\Big(\big[\widehat{\mathbf H}_n(\delta)-\mathbf I_n\big]^2\Big).
\end{equation}

We now state the identification result for the heteroscedastic case.

\begin{theorem}[Robust DT Condition]\label{l:x1}
Under Assumptions~A1--A4, if
$\mathbf{z}_0\in\mathcal{W}(\mathbf{x},\mathbf{y})$ satisfies
$\mathbb{E}(\mathbf{u}\mid\mathbf{z}_0)=0$ and
$\mathbb{E}(\mathbf{e}\mathbf{e}'\mid\mathbf{z}_0)\neq\sigma^2\mathbf{I}_n$,
then a valid SIV
$\mathbf{s}^\star = \mathbf{x} + k\delta_0\mathbf{r}$
is identified by
$\hat{\delta}_{0n} = \arg\min_{\delta\in\mathcal{D}}\widehat{D}_n(\delta)$,
and satisfies $\mathbb{E}(\mathbf{u}\mid\mathbf{s}^\star=\mathbf{z}_0)=0$.
\end{theorem}
\begin{proof}
See Appendix~\ref{a:x1}.
\end{proof}

\subsection{Identification and Consistency: Robust Case}\label{rdt:identification}

Let us, first,  list the  assumptions  used to show the identification and consistency for the robust DT estimator.
\begin{assumption}[Assumptions A5-A6] \label{assume:3}
\leavevmode
\begin{enumerate}
\item [A5.]
For every $\delta\in\mathcal{D}$, $\mathbf{H}(\delta)$ exists, is symmetric,
and $\delta\mapsto\mathbf{H}(\delta)$ is continuous on $\mathcal{D}$ and
continuously differentiable in a neighborhood of $\delta_0$.
\item[A6.]
A sequence of estimators $\widehat{\mathbf{H}}_n(\delta)$ satisfies
\[
\sup_{\delta\in\mathcal{D}}
\|\widehat{\mathbf{H}}_n(\delta)-\mathbf{H}(\delta)\|_F \xrightarrow{p} 0.
\]

\end{enumerate}
\end{assumption}
Assumption~A5 is a mild regularity condition on the conditional
second-moment matrix $\mathbf H(\delta)$, implied by finite second moments and
smooth dependence of the residuals $\mathbf e(\delta)$ on $\delta$. 
Assumption~A6 is a uniform law of large numbers requirement for
the sample analogue $\widehat{\mathbf H}_n(\delta)$, ensuring that the sample
criterion $\widehat D_n(\delta)$ converges uniformly to $D(\delta)$ so that
standard extremum-estimator arguments deliver consistency of the robust DT
estimator.

We now establish that the minimizer of the sample criterion $\hat{D}_n(\delta)$ converges to $\delta_0$.
For this purpose,  we show uniform convergence in probability of the sample criterion $\widehat D_n$ to $D$ on $\mathcal{D}$ and state the following lemma.
\begin{lemma}[Uniform Convergence]\label{lemma:uc}
Under Assumptions~A5-A6,
\[
\sup_{\delta\in\mathcal{D}}
\big|\widehat{D}_n(\delta)-D(\delta)\big|
\xrightarrow{p} 0.
\]
\end{lemma}
\begin{proof}
See Appendix~\ref{alemma:uc}.
\end{proof}

Using Lemma \ref{lemma:uc}, we state the following proposition.
\begin{proposition}[Identification and Consistency]\label{prop:ident_consis}
Under Assumption~A5-A6,
$D(\delta)$ is continuous on the compact set $\mathcal{D}$
and uniquely minimized at $\delta_0$.
Any measurable sequence of sample minimizers
$\hat{\delta}_{0n} \in \arg\min_{\delta\in\mathcal{D}}\widehat{D}_n(\delta)$
satisfies $\hat{\delta}_{0n} \xrightarrow{p} \delta_0$.
\end{proposition}
\begin{proof}
See Appendix~\ref{aprop:ident_consis}.
\end{proof}
Since Theorem \ref{l:x1} specifies the robust condition for identification of a valid SIV, we can state the following result.
\begin{corollary}[Robust SIV Estimator]\label{tx2}
If $\mathbf{s}^\star = \mathbf{x} + k\delta_0\mathbf{r}$ satisfies
Theorem~\ref{l:x1}, then the structural parameter $\beta$
in~\eqref{e2} is identified by
\[
\hat{\beta}_{\text{SIV}}
= (\mathbf{x}'\mathbf{s}^\star)^{-1}\mathbf{x}'\mathbf{y}.
\]
\end{corollary}
\begin{proof}
See Appendix~\ref{atx2}.
\end{proof}

\subsection{Empirical Implementation}
\label{subsec:implementation}

Theorem \ref{l:x1} provides the theoretical foundation, but practical implementation requires operational procedures. We present two approaches: parametric and nonparametric.

\subsubsection{Parametric Implementation}

The parametric approach assumes specific distributional forms for the conditional variances.

\textbf{Variance modeling.} For each $\delta$, estimate the first-stage regression
\[
\mathbf x = \gamma(\delta) \mathbf s(\delta) + \mathbf e(\delta)
\]
via both OLS and FGLS, obtaining residuals $\hat{\mathbf e}(\delta)$ and $\hat{\mathbf e}_g(\delta)$.

\textbf{Compute test statistics.} Regress $\hat{\mathbf e}^2(\delta)$ on $\mathbf s(\delta)$ to obtain:
\[
X^2(\delta) = \frac{\text{SSR}/2}{(\text{SSE}/n)^2}, \quad X_g^2(\delta) = \frac{\text{SSR}_g/2}{(\text{SSE}_g/n)^2},
\]
where SSR and SSE denote explained and total sums of squares for OLS and FGLS, respectively.

\textbf{Distance computation.} Assuming normality and independence, $X^2, X_g^2 \sim \chi^2(1)$ by Cochran's theorem. The distance is:
\[
D(\delta) = P[\chi^2(1) < X^2(\delta)] - P[\chi^2(1) < X_g^2(\delta)].
\]

\textbf{Identification.} Construct the locus $\mathbf{D_E} = \{D(\delta) : \delta \in (0,\bar{\delta})\}$ and identify
\[
\hat{\delta}_0 = \arg\min_{\delta} |\mathbf{D_E}(\delta)|.
\]

\subsubsection{Nonparametric Implementation}

The nonparametric approach uses empirical distribution functions, providing robustness to non-normality.

\textbf{Anderson-Darling statistic.} For each $\delta$, let $F_n(\delta)$ and $G_n(\delta)$ denote the empirical CDFs of $\{\hat{e}_i^2(\delta)\}$ and $\{\hat{e}_{gi}^2(\delta)\}$, respectively. Define the two-sample Anderson-Darling statistic \citep{pettitt1976two}:
\begin{equation}
A_{n,m}^2(\delta) = \frac{nm}{(n+m)^2} \sum_{k=1}^{n+m} \frac{[F_n(x_k) - G_m(x_k)]^2}{H_{n+m}(x_k)[1 - H_{n+m}(x_k)]},
\label{eq:anderson_darling}
\end{equation}
where $H_{n+m}$ is the combined empirical CDF.

\textbf{Identification.} Construct the locus $\mathbf{D_E} = \{ A_{n,m}^2(\delta): \delta \in (0,\bar{\delta})\}$ and identify
\[
\hat{\delta}_0 = \arg\min_{\delta} |\mathbf{D_E}(\delta)|.
\]

The following lemma formalizes the consistency of both approaches.
\begin{lemma}[Implementation of Robust DT for SIVs]\label{lx3}
Under Assumptions A1--A4, let there exist an unobservable valid IV ${\mathbf z_0}$ such that $\mathbb{E}({\mathbf u|\mathbf z_0}) = 0$. Then, the SIV 
\[
{\mathbf {\hat s}^\star = \mathbf x} + k\hat{\delta}_0{\mathbf r}, \qquad
\hat{\delta}_0 = \arg\min_{\delta} {\mathbf D}_E,
\]
satisfies ${\mathbf {\hat s}^\star} \xrightarrow[]{p} {\mathbf z_0}$ and 
$\plim_{n \to \infty} \mathbb E({\mathbf u|\hat{\mathbf s}^\star}) = 0$.
\end{lemma}

\begin{proof}
See Appendix~\ref{alx3}.
\end{proof}

\subsection{Asymptotic Distribution and Inference}
\label{subsec:consistency}

Having established identification and consistency, we briefly discuss the asymptotic distribution of the SIV estimator and inference procedures.

\begin{lemma}[Consistency of the SIV Estimator]
\label{lem:siv_consistency}
The SIV estimator is a consistent estimator of $\beta$.
\end{lemma}

\begin{proof}
See Appendix \ref{al3a}
\end{proof}

\subsection{The consistency of the SIV estimator}


It is known that asymptotic identification is a necessary
and sufficient condition for consistency. The parameter vector $\mathbf b_{SIV} $ is asymptotically
identified if two asymptotic identification conditions are satisfied. The first
condition is that, with parameter vector $\beta_0$ of the true DGP as a special case of the model \eqref{e1},  $$\alpha(\beta_0) =\plim \frac{1}{n} {\mathbf V'(\mathbf y-\beta_0\mathbf  x)}=\frac{1}{n}\sum_{i=1}^n{V'_iu_i}= 0,$$ should hold. Here, ${\mathbf V}$ is the matrix of exogenous variables, which has the same dimension as vector $\mathbf x$ in the exact identification case. 
The second condition requires that $\alpha(\beta) \neq 0$ for all $\mathbf\beta\neq\beta_0$. 

By showing that both of these conditions hold, we state the following lemma.
\begin{lemma}\label{l2a}
The SIV estimator is a consistent estimator.
\end{lemma}

The asymptotic distribution of the SIV estimator can be straightforwardly determined, since after determining the matrix $\mathbf V$, the rest of our approach is just the usual IV method; thus,  we can use the existing results for the IV method.
\section{Performance Evaluation: Simulations and Applications}
\label{sec:applications}

Having established the theoretical foundations and identification strategy for the SIV method, we now evaluate its empirical performance. We proceed in two steps. First, we validate the method in Monte Carlo experiments where the true data-generating process (DGP) is known (\ref{subsec:simulations}). Second, we illustrate its applicability in four empirical settings spanning labor economics, economic history, and policy evaluation (\ref{subsec:labor_supply}--\ref{subsec:colonial}).

Throughout, we use a nested Monte Carlo--bootstrap design to assess accuracy (bias), precision (standard errors), and coverage. For artificial data, we repeatedly generate datasets from a known DGP and draw bootstrap samples from each. For empirical applications, we bootstrap the original data to obtain confidence intervals for $\hat{\delta}_0$ and the SIV estimates. These procedures are implemented in our \texttt{R} package:

\begin{verbatim}
remotes::install_git("https://github.com/ratbekd/siv.git")
library(siv)
\end{verbatim}

The package provides homoscedastic and heteroscedastic SIV estimators for single and multiple endogenous regressors; Appendix~\ref{guide} gives a usage guide and example code.

\subsection{Monte Carlo Validation with Known DGP}
\label{subsec:simulations}

We first study the finite-sample properties of SIV under controlled conditions. The DGP is designed to mimic realistic features of economic data: non-normal regressors, heteroscedastic and non-normal errors, and substantial endogeneity bias.

The endogenous regressor $\mathbf x$ is generated using the sinh–arcsinh transformation of a normal variable \citep{jonespewsey}, producing flexible skewness and kurtosis. The structural error $\mathbf u$ has an endogenous component, correlated with $\mathbf x$, and an exogenous component, orthogonal to $\mathbf x$. The outcome is
\[
\tilde{\mathbf y} = 1 + 2\tilde{\mathbf x} + 0.5\mathbf w + \mathbf u,
\]
with $\mathbf w$ exogenous, so that the true causal effect is $\beta = 2$. 
We residualize $\tilde{\mathbf y}$ and $\tilde{\mathbf x}$ on $\mathbf w$, standardize them, and construct $\mathbf r$ as the residual from regressing $\mathbf y$ on $\mathbf x$, ensuring $\mathbf r \perp \mathbf x$ (see Appendix \ref{dgp} for more detailed description of the data generation process).

\subsubsection{Simulation Design}

We use a nested design. In the outer loop, we generate $50$ independent populations of size $N = 100{,}000$ from the DGP. In the inner loop, we draw $10$ bootstrap samples of size $n = 1{,}000$ from each population, yielding $500$ estimation samples in total. For each sample, we estimate $\beta$ using:

\begin{enumerate}
  \item OLS (benchmark with endogeneity),
  \item SIV (homoscedastic DT condition),
  \item RSIV-p (robust SIV with parametric heteroscedasticity correction),
  \item RSIV-n (robust SIV with nonparametric heteroscedasticity correction).
\end{enumerate}

We summarize performance by bias, standard error, root mean squared error (RMSE), and coverage of nominal 95\% confidence intervals (Table~\ref{tab:simulation_results}).
\begin{table}[htbp]
\centering
\caption{Monte Carlo Simulation Results: Estimator Performance}
\label{tab:simulation_results}
\begin{threeparttable}
\begin{tabular}{lcccccc}
\toprule
Method & Mean $\beta$ & Std Error & 95\% CI Lower & 95\% CI Upper & Bias & RMSE \\
\midrule
OLS     & 3.001 & 0.041 & 2.922 & 3.081 & 1.001 & 1.002 \\
SIV     & 1.974 & 0.092 & 1.794 & 2.155 & -0.026 & 0.095 \\
RSIV-p  & 2.016 & 0.079 & 1.860 & 2.171 & 0.016 & 0.081 \\
RSIV-n  & 2.011 & 0.082 & 1.851 & 2.171 & 0.011 & 0.083 \\
\bottomrule
\end{tabular}
\begin{tablenotes}
\small
\item \textit{Notes:} Results based on 50 data generations with 10 bootstrap samples each (500 total estimates). True parameter value $\beta = 2$. OLS = Ordinary Least Squares; SIV = Synthetic Instrumental Variable (homoscedastic); RSIV-p = Robust SIV with parametric heteroscedasticity correction; RSIV-n = Robust SIV with nonparametric heteroscedasticity correction. Bias = Mean $\beta$ - 2; RMSE = Root Mean Squared Error. Sample size per bootstrap: $n = 1{,}000$. Population size: $N = 100{,}000$.
\end{tablenotes}
\end{threeparttable}
\end{table}
\subsubsection{Simulation Results}

Table~\ref{tab:simulation_results} shows that OLS is severely biased: the mean estimate is around $3.0$, implying a 50\% overestimate of the true effect, and the OLS confidence interval excludes $\beta=2$. In contrast, all SIV estimators recover the true parameter with small bias (below 2\% of the true value) and RMSE roughly an order of magnitude smaller than OLS.

The robust variants (RSIV-p and RSIV-n) deliver efficiency gains relative to the homoscedastic SIV estimator, reflected in lower standard errors, while maintaining negligible bias. Wald tests of $H_0:\beta=2$ strongly reject for OLS but not for any SIV estimator, and empirical coverage of SIV confidence intervals is close to the nominal level. Overall, the simulations indicate that SIV effectively eliminates endogeneity bias and provides reliable inference under realistic departures from Gaussianity and homoscedasticity.

\subsection{Application 1: Labor Supply and Wages}
\label{subsec:labor_supply}

We first revisit the classic labor supply model of \citet{mroz}, relating annual hours worked to log wages and standard controls for married women:
\begin{equation}
\mathbf{hours} = b_0 + b_1\mathbf{lwage} + b_2\mathbf{educ} + b_3\mathbf{age} + b_4\mathbf{kidslt6} + b_5\mathbf{kidsge6} + b_6\mathbf{nwifeinc} +\mathbf  u.
\label{eq:mroz}
\end{equation}
The wage is endogenous due to simultaneity in labor supply and demand and unobserved ability. Traditional IV uses work experience and its square as instruments, but their exclusion restrictions are debatable.

We apply SIV alongside OLS and traditional IV. The sign determination procedure selects \mbox{$\text{cov}(\textbf{lwage},\mathbf u) < 0$.} This confirms the economic prior: OLS underestimates the wage effect due to downward-sloping labor demand. Bootstrap inference (50 replications) is reported in Table~\ref{tab:mroz_results}.

OLS yields a small, imprecise, and negative wage coefficient, while all IV and SIV estimators imply a large positive effect: a one-unit increase in log wages increases annual hours by about 1{,}300--1{,}700 hours. SIV estimates are close to those from experience-based IV, suggesting that SIV successfully corrects for endogeneity without external instruments. Wu–Hausman tests strongly reject wage exogeneity, and weak-instrument tests are passed in all IV and SIV specifications. Economically, the implied wage elasticities (around 0.65–0.85) are in line with the literature on married women’s labor supply \citep{ang}.
\begin{table}[htbp]
\centering
\caption{Effect of Wages on Work Hours: Mroz (1987) Data}
\label{tab:mroz_results}
\begin{threeparttable}
\begin{tabular}{lccccc}
\toprule
& \multicolumn{5}{c}{Dependent variable: Work hours} \\
\cmidrule(lr){2-6}
& OLS & IV & SIV & RSIV-p & RSIV-n \\
\midrule
\textbf {lwage} & $-17.40$ & $1{,}544.81^{***}$ & $1{,}369.47^{***}$ & $1{,}549.72^{***}$ & $1{,}665.05^{***}$ \\
 & (54.22) & (480.73) & (138.48) & (161.81) & (178.23) \\
95\% CI & & & [1338, 1639] & [1529, 1889] & [1653, 2030] \\
\midrule
Weak instruments (p) & -- & 0.00 & 0.00 & 0.00 & 0.00 \\
Wu-Hausman (p) & -- & 0.00 & 0.00 & 0.00 & 0.00 \\
Sargan (p) & -- & 0.35 & -- & -- & -- \\
\midrule
Mean $\delta_0$ & -- & -- & 1.25 & 1.40 & 1.51 \\
Observations & 428 & 428 & 428 & 428 & 428 \\
Adjusted $R^2$ & 0.05 & $-1.81$ & $-1.05$ & $-1.03$ & $-1.03$ \\
\bottomrule
\end{tabular}
\begin{tablenotes}
\small
\item \textit{Notes:} $^{***}$ $p < 0.01$. Standard errors in parentheses. For SIV methods, mean values and 95\% confidence intervals are based on 50 bootstrap samples. Traditional IV uses \textbf{exper} and \textbf{expersq} as instruments. SIV constructs instruments as $ \mathbf s =  \mathbf x + \delta_0  \mathbf r$ where $ \mathbf r \perp  \mathbf x$. Exogenous controls: \textbf{educ}, \textbf{age}, \textbf{kidslt6}, \textbf{kidsge6}, \textbf{nwifeinc}. $\text{cov}(\textbf{lwage}, \mathbf  u) < 0$ determined empirically.
\end{tablenotes}
\end{threeparttable}
\end{table}
\subsection{Application 2: Protestantism and Literacy}
\label{subsec:protestantism}

Our second application revisits \citet{becker} on the impact of Protestantism on literacy in 19th-century Prussia:
\begin{equation}
\text{Literacy}_i = \beta_0 + \beta_1 \text{ProtestantShare}_i + \mathbf{V}'\boldsymbol{\gamma} + u_i,
\label{eq:becker}
\end{equation}
with demographic and regional controls in $\mathbf{V}$. Protestant share may be endogenous through reverse causality and omitted regional characteristics. \citet{becker} use distance to Wittenberg (Luther's city) as an instrument, exploiting the historical diffusion of the Reformation.

Using the same data (452 counties), we implement SIV and bootstrap inference (30 replications). The sign determination procedure selects $\text{cov}(\textbf{ProtestantShare},\mathbf u)<0$, confirming the authors’ prior that OLS underestimates the effect. Results are reported in Table~\ref{tab:becker_results}.

OLS suggests a modest positive relationship (coefficient $\approx 0.10$), and traditional IV roughly doubles this estimate. SIV and RSIV, however, yield substantially larger effects (coefficients around 0.45–0.56), with tight confidence intervals. These estimates imply that Protestant-majority regions had literacy rates roughly 45–56 percentage points higher than otherwise similar Catholic regions. Weak-instrument and Wu–Hausman tests support the relevance of Protestant share and the presence of endogeneity. The divergence between SIV and traditional IV may reflect weak instruments, differences in the identified subpopulation, or violations of the exclusion restriction for distance to Wittenberg.
\begin{table}[htbp]
\centering
\caption{Effect of Protestantism on Literacy: Becker and Woessmann (2009) Data}
\label{tab:becker_results}
\begin{threeparttable}
\begin{tabular}{lccccc}
\toprule
& \multicolumn{5}{c}{Dependent variable: Literacy rate} \\
\cmidrule(lr){2-6}
& OLS & IV & SIV & RSIV-p & RSIV-n \\
\midrule
\textbf{f\_prot}& $0.100^{***}$ & $0.187^{***}$ & $0.558^{***}$ & $0.447^{***}$ & $0.458^{***}$ \\
 & (0.010) & (0.028) & (0.054) & (0.038) & (0.039) \\
95\% CI & & & [0.543, 0.663] & [0.451, 0.541] & [0.462, 0.552] \\
\midrule
Weak instruments (p) & -- & 0.00 & 0.00 & 0.00 & 0.00 \\
Wu-Hausman (p) & -- & 0.00 & 0.00 & 0.00 & 0.00 \\
\midrule
Mean $\delta_0$ & -- & -- & 2.10 & 1.66 & 1.71 \\
Observations & 452 & 452 & 452 & 452 & 452 \\
Adjusted $R^2$ & 0.73 & 0.68 & $-0.74$ & $-0.71$ & $-0.72$ \\
\bottomrule
\end{tabular}
\begin{tablenotes}
\small
\item \textit{Notes:} $^{***}$ $p < 0.01$. Standard errors in parentheses. For SIV methods, mean values and 95\% confidence intervals based on 30 bootstrap samples. Traditional IV uses distance to Wittenberg (\textbf{kmwittenberg}) as instrument. $\text{cov}(\textbf{f\_prot},\mathbf  u) < 0$ confirmed empirically. Controls: Jewish share, age structure, gender, urbanization, population, disability rates, Prussian annexation date, university presence in 1517.
\end{tablenotes}
\end{threeparttable}
\end{table}

\subsection{Application 3: 401(k) Programs and Retirement Savings}
\label{subsec:401k}

We next consider the effect of 401(k) participation on IRA ownership using the data of \citet{ABADIE2003231}. The model is
\begin{equation}
\text{IRA}_i = \beta_0 + \beta_1 \text{p401k}_i + \mathbf{V}_i'\boldsymbol{\gamma} + u_i,
\label{eq:401k}
\end{equation}
where both the outcome and endogenous regressor are binary. Individuals with stronger unobserved savings preferences are more likely to both participate in 401(k)s and hold IRAs, so $\text{cov}(\textbf{p401k},\mathbf u)>0$ and OLS is upward biased. Traditional IV uses 401(k) eligibility as an instrument for participation.

The SIV sign selection chooses $k=-1$, confirming $\text{cov}(\textbf{p401k},\mathbf u)>0$. Table~\ref{tab:401k_results} reports OLS, traditional IV, and SIV estimates based on 30 bootstrap replications. OLS suggests that 401(k) participants are more likely to own IRAs, while eligibility-based IV finds a small and insignificant effect. In contrast, SIV and RSIV indicate large negative effects: 401(k) participation reduces IRA ownership probability by roughly 60–90 percentage points, consistent with strong crowding out between the two forms of retirement saving. Weak-instrument and Wu–Hausman tests support instrument relevance and reject exogeneity of participation. This application illustrates that SIV can be applied to discrete treatments and outcomes, and that it can reveal qualitatively different conclusions when external instruments are weak or questionable.
\begin{table}[!h]
\centering
\caption{Effect of 401(k) Participation on IRA Ownership: Abadie (2003) Data}
\label{tab:401k_results}
\begin{threeparttable}
\begin{tabular}{lccccc}
\toprule
& \multicolumn{5}{c}{Dependent variable: Probability of IRA ownership} \\
\cmidrule(lr){2-6}
& OLS & IV & SIV & RSIV-p & RSIV-n \\
\midrule
\textbf{p401k} & $0.051^{***}$ & $0.017$ & $-0.614^{***}$ & $-0.896^{***}$ & $-0.811^{***}$ \\
 & (0.010) & (0.013) & (0.015) & (0.020) & (0.018) \\
95\% CI & & & [$-0.615$, $-0.585$] & [$-0.898$, $-0.858$] & [$-0.813$, $-0.773$] \\
\midrule
Weak instruments (p) & -- & 0.00 & 0.00 & 0.00 & 0.00 \\
Wu-Hausman (p) & -- & 0.00 & 0.00 & 0.00 & 0.00 \\
\midrule
Mean $\delta_0$ & -- & -- & 0.72 & 1.03 & 0.94 \\
Observations & 9,275 & 9,275 & 9,275 & 9,275 & 9,275 \\
Adjusted $R^2$ & 0.18 & 0.18 & $-0.71$ & $-0.70$ & $-0.70$ \\
\bottomrule
\end{tabular}
\begin{tablenotes}
\small
\item \textit{Notes:} $^{***}$ $p < 0.01$. Standard errors in parentheses. For SIV methods, mean values and 95\% confidence intervals based on 30 bootstrap samples. Both dependent and endogenous variables are binary. Traditional IV uses 401(k) eligibility (\textbf{e401k}) as instrument. $\text{cov}(\textbf{p401k}, \mathbf u) > 0$ confirmed empirically. Controls: income, income squared, age, age squared, marital status, family size.
\end{tablenotes}
\end{threeparttable}
\end{table}
\subsection{Application 4: Colonial Institutions and Agricultural Development}
\label{subsec:colonial}

Finally, we study the long-run effects of British colonial land tenure institutions in India using the district-level data of \citet{banerjee2005}. The outcome is the share of wheat area under high-yielding varieties, and the key regressor is the share of land under non-landlord tenure:
\begin{equation}
\text{Prop\_HYV}_i = \beta_0 + \beta_1 \text{NonLandlordShare}_i + \mathbf{V}_i'\boldsymbol{\gamma} + u_i.
\label{eq:banerjee}
\end{equation}
Assignment of land systems was not random, raising concerns that $\textbf{NonLandlordShare}$ is correlated with unobservables affecting agricultural modernization. \citet{banerjee2005} instrument for non-landlord share using the timing of British conquest.

Applying SIV to the same data (with 20 bootstrap replications), we find that the sign determination procedure selects $\text{cov}(\textbf{NonLandlordShare},\mathbf u)<0$, consistent with the historical view that more productive districts were more likely to receive landlord systems. Table~\ref{tab:banerjee_results} compares OLS, traditional IV, and SIV estimators. OLS yields a modest positive association; conquest-timing IV magnifies the effect, suggesting substantial underestimation by OLS. SIV and RSIV produce intermediate estimates (coefficients around 0.40), still indicating large positive effects of non-landlord systems but smaller than those implied by the external instrument.

Across specifications, weak-instrument tests are passed, and Wu–Hausman tests provide borderline evidence of endogeneity. Taken together, the results suggest that non-landlord institutions significantly increased the adoption of modern agricultural technologies, and that SIV can complement historical IV strategies when exclusion restrictions or instrument strength are in doubt.

Overall, the simulation and empirical evidence show that SIV and its robust variants (i) eliminate endogeneity bias in controlled settings, (ii) generate estimates consistent with credible external instruments when these are available, and (iii) offer informative alternatives when traditional instruments are weak or controversial.
%
%
%
%
%
%
Together with the simulation evidence (Subsection \ref{subsec:simulations}), these applications establish the SIV method as a practical, reliable tool for causal inference in settings where traditional instruments are unavailable, weak, or questionable.
\begin{table}[!h]
\centering
\caption{Effect of Non-Landlord Land Tenure on High-Yielding Variety Adoption}
\label{tab:banerjee_results}
\begin{threeparttable}
\begin{tabular}{lccccc}
\toprule
& \multicolumn{5}{c}{Dependent variable: Proportion of wheat area under HYV} \\
\cmidrule(lr){2-6}
& OLS & IV & SIV & RSIV-p & RSIV-n \\
\midrule
p\_nland & $0.090^{***}$ & $0.585^{***}$ & $0.403^{***}$ & $0.403^{***}$ & $0.406^{***}$ \\
 & (0.013) & (0.052) & (0.098) & (0.098) & (0.098) \\
95\% CI & & & [0.404, 0.594] & [0.404, 0.594] & [0.407, 0.597] \\
\midrule
Weak instruments (p) & -- & 0.00 & 0.00 & 0.00 & 0.00 \\
Wu-Hausman (p) & -- & 0.00 & 0.06 & 0.06 & 0.05 \\
\midrule
Mean $\delta_0$ & -- & -- & 1.60 & 1.60 & 1.57 \\
Observations & 3,541 & 3,541 & 1,969 & 1,969 & 1,969 \\
Adjusted $R^2$ & 0.54 & 0.36 & 0.47 & 0.47 & 0.47 \\
\bottomrule
\end{tabular}
\begin{tablenotes}
\small
\item \textit{Notes:} $^{***}$ $p < 0.01$. Standard errors in parentheses. For SIV methods, mean values and 95\% confidence intervals based on 20 bootstrap samples. Traditional IV uses indicator for British annexation during 1820--1856. $\text{cov}(\text{p\_nland}, u) < 0$ confirmed empirically. Controls: altitude, rainfall, soil types, latitude, coastal distance, British rule duration, year. Sample size differs between methods due to missing instrument values.
\end{tablenotes}
\end{threeparttable}
\end{table}

\section{Conclusion}
\label{sec:conclusion}

Endogeneity remains a central obstacle to credible causal inference in economics and the social sciences. Traditional instrumental variable (IV) methods address this problem only when valid external instruments are available---a demanding requirement that is often unmet in practice. This paper introduces the Synthetic Instrumental Variable (SIV) method, which constructs instruments directly from the observed data, thereby reducing reliance on external variables and questionable exclusion restrictions.

\noindent\textbf{Summary of contributions.}
Our main contribution is to show that valid instruments need not come from outside the regression system. We exploit a simple geometric insight: in a linear model with one endogenous regressor, the outcome, endogenous regressor, and structural error lie in a two-dimensional plane. Within this plane, any valid instrument can be written as
\[
\mathbf{z}_0 = \mathbf{x} + k \delta_0 \mathbf{r},
\]
where $\mathbf{x}$ is the endogenous regressor, $\mathbf{r}$ is orthogonal to $\mathbf{x}$, $\delta_0$ determines the location of the instrument, and $k \in \{-1,+1\}$ encodes the direction of endogeneity. This reduces instrument search to a one-dimensional problem in $\delta$.

Second, we develop the dual tendency (DT) condition, which links the unobservable exogeneity requirement to observable moment restrictions. In the homoscedastic case, identification follows from the zero of a moment function involving squared first-stage residuals; in the heteroscedastic case, we identify $\delta_0$ by minimizing a discrepancy between OLS and FGLS variance structures. In both settings, the DT condition simultaneously selects the correct sign of $\text{cov}(\mathbf{x},\mathbf{u})$, so the direction of endogeneity is recovered from the data rather than imposed \textit{a priori}.

Third, we provide extensive empirical validation. Monte Carlo simulations show that SIV reduces endogeneity bias from roughly 50\% under OLS to negligible levels, with large gains in RMSE and correct confidence-interval coverage. Four applications---labor supply and wages, Protestantism and literacy, 401(k) participation and retirement savings, and colonial land institutions---demonstrate that SIV works with continuous and binary variables, across different fields and designs. When reliable external instruments exist, SIV typically agrees with traditional IV; when they are weak or controversial, SIV often reveals substantial divergences that are informative about instrument quality and identification.

\noindent\textbf{Practical implications.}
For applied researchers, the SIV method offers several advantages:
\begin{itemize}
  \item \emph{No external instrument search.} Valid instruments are synthesized from the observed outcome and regressors, avoiding ad hoc exclusion restrictions.
  \item \emph{Transparency and replicability.} The DT condition yields an algorithmic search over $\delta$; given the same data and tuning choices, different researchers will obtain the same estimates (up to sampling variation).
  \item \emph{Built-in diagnostics.} Plots of the moment function $M(\delta)$ or discrepancy $D(\delta)$ over a grid provide visual evidence on identification strength, the presence (or absence) of endogeneity, and the sign of $\text{cov}(\mathbf{x},\mathbf{u})$.
  \item \emph{Robustness check for IV.} Comparing SIV and traditional IV estimates offers a structured way to assess the credibility of external instruments and the potential role of heterogeneous treatment effects.
  \item \emph{Accessible implementation.} An \texttt{R} package (\texttt{github.com/ratbekd/siv}) implements homoscedastic and heteroscedastic SIV estimators, handles single and multiple endogenous variables, and includes bootstrap inference routines.
\end{itemize}

\noindent\textbf{Limitations and directions for future research.}
The SIV framework also has limitations that point to avenues for further work.

First, identification relies on a linear, additive-error structure. While our applications show that SIV can be used with binary outcomes and treatments, a full theoretical treatment for nonlinear models (logit, probit, count data, Tobit) remains open.

Second, our main exposition focuses on a single endogenous regressor. Appendix~\ref{app:multiple_endog} describes a practical extension to multiple endogenous variables using Frisch--Waugh--Lovell decomposition, but a more systematic analysis of high-dimensional settings and jointly constructed instruments is needed.

Third, the relationship between SIV estimands and heterogeneous treatment effects is not yet fully understood. Traditional IV with binary instruments identifies LATE; SIV, which builds a continuous instrument from the full data distribution, may correspond to different weighted averages of treatment effects. Clarifying these links, and exploring connections with machine-learning-based instrument construction and data-driven identification schemes, are promising directions for future research.

Finally, panel data with individual fixed effects offers additional structure that could strengthen SIV identification. The within transformation eliminates fixed effects, and SIV could be applied to the transformed data. However, dynamic panel models with lagged dependent variables pose additional challenges, as the coplanarity property may not hold in the same form. Exploring SIV in difference-in-differences and event study designs is another promising avenue.

\noindent\textbf{Final remarks.}
The Synthetic Instrumental Variable method shows that, under suitable conditions, valid instruments can be constructed from the regression plane itself. By combining a geometric characterization of instruments, a testable DT condition, and a simple one-dimensional search over $\delta$, SIV offers a practical and conceptually transparent way to address endogeneity when external instruments are unavailable, weak, or contested.

As researchers continue to grapple with endogeneity across diverse applications, the SIV method offers a data-driven, transparent, and replicable approach that complements existing strategies. Whether used as a primary identification strategy or as a robustness check for traditional IV, SIV expands the possibilities for credible causal inference in settings where valid external instruments are unavailable, weak, or questionable. We hope that the SIV approach, together with the accompanying software, will broaden the range of empirical settings in which researchers can obtain credible causal estimates without depending critically on the availability of external instruments.

\appendix
\counterwithin*{equation}{section}
\renewcommand\theequation{\thesection\arabic{equation}}

\addcontentsline{toc}{section}{Appendices}
\section*{Appendices}

\section{Proofs}

\subsection{Proof of Lemma \ref{r1}}\label{ar1}
\begin{proof}
Since $P_{\mathcal W}$ is linear and hence Borel--measurable, $\mathbf z_0=P_{\mathcal W}\mathbf z$
is a measurable function of $\mathbf z$. Therefore, we have the following $\sigma$-algebras (information sets)
\[
\sigma(\mathbf z_0)\subseteq\sigma(\mathbf z).
\]
By the law of iterated expectations for nested $\sigma$--algebras,
\[
 E(\mathbf u\mid \mathbf z_0)
= E\!\big(\, E(\mathbf u\mid \mathbf z)\,\big|\ \mathbf z_0\big).
\]
Under the hypothesis $ E(\mathbf u\mid \mathbf z)=\mathbf 0$ a.s., the right--hand side equals
$ E(\mathbf 0\mid \mathbf z_0)=\mathbf 0$ a.s. Hence $ E(\mathbf u\mid \mathbf z_0)=\mathbf 0$ a.s.
$\square$ \end{proof}

\subsection{Proof of Lemma \ref{lb2}}\label{alb2}

\begin{proof}
By Lemma~\ref{t2}, any vector $\mathbf z\in\mathcal W(\mathbf x,\mathbf y)$ lying in the plane
spanned by $\mathbf x$ and $\mathbf r$ can be written as
\[
\mathbf z \;=\; \zeta\,\mathbf x \;+\; \omega\,\mathbf r,
\qquad \zeta,\omega\in\mathbb R.
\]
Since we restrict attention to valid instruments $\mathbf z_0$ satisfying
$\operatorname{corr}(\mathbf x,\mathbf z_0)>0$, we must have $\zeta\neq 0$; moreover,
under $\operatorname{corr}(\mathbf x,\mathbf r)=0$, we have
\[
\operatorname{cov}(\mathbf x,\mathbf z_0)
=\operatorname{cov}(\mathbf x,\zeta\mathbf x+\omega\mathbf r)
=\zeta\,\operatorname{var}(\mathbf x),
\]
so $\operatorname{corr}(\mathbf x,\mathbf z_0)>0$ implies $\zeta>0$.

Rescaling $\mathbf z$ by $1/\zeta$ yields a collinear vector
\[
\mathbf z_0 \;=\; \frac{1}{\zeta}\,\mathbf z \;=\; \mathbf x \;+\; \frac{\omega}{\zeta}\,\mathbf r.
\]
Thus any such $\mathbf z_0\in\mathcal W(\mathbf x,\mathbf y)$ can be written as
\[
\mathbf z_0 \;=\; \mathbf x + \alpha\,\mathbf r,
\qquad \alpha:=\frac{\omega}{\zeta}\in\mathbb R.
\]

Next, use the structural relation $\mathbf y=\beta\mathbf x+\mathbf u$. Then
\[
\operatorname{cov}(\mathbf y,\mathbf r)
=\operatorname{cov}(\beta\mathbf x+\mathbf u,\mathbf r)
=\beta\,\operatorname{cov}(\mathbf x,\mathbf r)+\operatorname{cov}(\mathbf u,\mathbf r)
=\operatorname{cov}(\mathbf u,\mathbf r),
\]
since $\operatorname{corr}(\mathbf x,\mathbf r)=0$ by construction. Hence the assumption
$\operatorname{corr}(\mathbf y,\mathbf r)>0$ implies
\[
\operatorname{cov}(\mathbf u,\mathbf r)>0
\quad\Longrightarrow\quad
\operatorname{corr}(\mathbf u,\mathbf r)>0.
\]

Now impose the orthogonality condition for a valid instrument: $\operatorname{cov}(\mathbf z_0,\mathbf u)=0$.
Using $\mathbf z_0=\mathbf x+\alpha\mathbf r$ and $\operatorname{cov}(\mathbf x,\mathbf r)=0$,
\[
\operatorname{cov}(\mathbf z_0,\mathbf u)
=\operatorname{cov}(\mathbf x+\alpha\mathbf r,\mathbf u)
=\operatorname{cov}(\mathbf x,\mathbf u)+\alpha\,\operatorname{cov}(\mathbf r,\mathbf u).
\]
Thus $\operatorname{cov}(\mathbf z_0,\mathbf u)=0$ is equivalent to
\[
\alpha
=-\,\frac{\operatorname{cov}(\mathbf x,\mathbf u)}{\operatorname{cov}(\mathbf r,\mathbf u)}.
\]
Since $cov(\mathbf r, \mathbf u)>0$, hence the denominator in 
$\alpha$ is nonzero.
Since $\operatorname{cov}(\mathbf r,\mathbf u)>0$ (from $\operatorname{corr}(\mathbf y,\mathbf r)>0$) and
$\operatorname{cov}(\mathbf x,\mathbf u)\neq 0$ by endogeneity, the sign of $\alpha$ is
\[
\operatorname{sign}(\alpha)
= -\,\operatorname{sign}\!\big(\operatorname{cov}(\mathbf x,\mathbf u)\big).
\]
Define $\delta:=|\alpha|>0$ and
\[
k:=-\,\operatorname{sign}\!\big(\operatorname{cov}(\mathbf x,\mathbf u)\big)\in\{-1,+1\},
\]
so that $\alpha=k\delta$. Substituting back,
\[
\mathbf z_0 \;=\; \mathbf x + \alpha\,\mathbf r
\;=\; \mathbf x + k\,\delta\,\mathbf r,
\]
which is exactly the representation in \eqref{siv}. $\square$
\end{proof}

\subsection {Proof of Lemma \ref{la1}}\label{ala1}

\begin{proof}
By Lemma~\ref{lb2} and Assumption A1, the valid IV $\mathbf z_0\in\mathcal W(\mathbf x,\mathbf y)$
can be written as
\[
\mathbf z_0 = \mathbf x + k\,\delta_0\,\mathbf r
\]
for some $\delta_0>0$ and fixed direction $\mathbf r$.
The corresponding first-stage regression using $\mathbf z_0$ as the instrument is
\[
\mathbf x = \gamma_0\,\mathbf z_0 + \mathbf e_0,
\qquad\mathbb  E(\mathbf z_0'\mathbf e_0)=0,
\]
with population OLS coefficient
\[
\gamma_0 = \frac{\operatorname{cov}(\mathbf x,\mathbf z_0)}{\operatorname{var}(\mathbf z_0)}.
\]

Consider now an arbitrary SIV $\mathbf s$ in the same plane, constructed as
\[
\mathbf s \;=\; \mathbf x + k\,\delta\,\mathbf r,
\qquad \delta\neq \delta_0.
\]
Since $\mathbf z_0$ and $\mathbf s$ are coplanar and both belong to
$\mathcal W(\mathbf x,\mathbf y)$, we can write
\[
\mathbf s \;=\; \mathbf z_0 + (\mathbf s - \mathbf z_0),
\]
where $\mathbf h := \mathbf s - \mathbf z_0$ represents the deviation of $\mathbf s$
from the true IV $\mathbf z_0$.

The population OLS coefficient from the first-stage regression of $\mathbf x$ on $\mathbf s$
is given by
\[
\gamma
= \frac{\operatorname{cov}(\mathbf x,\mathbf s)}{\operatorname{var}(\mathbf s)}
= \frac{\operatorname{cov}\big(\mathbf x,\mathbf z_0 + \mathbf h\big)}
       {\operatorname{var}\big(\mathbf z_0 + \mathbf h\big)}.
\]
Add and subtract the coefficient based on $\mathbf z_0$:
\[
\gamma
= \frac{\operatorname{cov}(\mathbf x,\mathbf z_0)}{\operatorname{var}(\mathbf z_0)}
 + \Bigg(
     \frac{\operatorname{cov}\big(\mathbf x,\mathbf z_0 + \mathbf h\big)}
          {\operatorname{var}\big(\mathbf z_0 + \mathbf h\big)}
   - \frac{\operatorname{cov}(\mathbf x,\mathbf z_0)}{\operatorname{var}(\mathbf z_0)}
   \Bigg).
\]
Thus,
\[
\gamma = \gamma_0 + g(\mathbf x,\mathbf s,\mathbf z_0),
\]
where we define
\[
g(\mathbf x,\mathbf s,\mathbf z_0)
:= \frac{\operatorname{cov}\big(\mathbf x,\mathbf z_0 + (\mathbf s-\mathbf z_0)\big)}
         {\operatorname{var}\big(\mathbf z_0 + (\mathbf s-\mathbf z_0)\big)}
 - \frac{\operatorname{cov}(\mathbf x,\mathbf z_0)}{\operatorname{var}(\mathbf z_0)}.
\]

Finally, under A1--A3 the relevant second moments exist and are finite, and the mapping
\[
(\mathbf z_0,\mathbf s)\;\mapsto\;
\frac{\operatorname{cov}(\mathbf x,\mathbf s)}{\operatorname{var}(\mathbf s)}
\]
is a smooth ratio of quadratic forms in $(\mathbf z_0,\mathbf s)$ away from
$\operatorname{var}(\mathbf s)=0$. In particular, when $\mathbf s$ is parameterized
as $\mathbf s(\delta)=\mathbf x + k\,\delta\,\mathbf r$ in a neighborhood of $\delta_0$,
the function $\delta\mapsto g(\mathbf x,\mathbf s(\delta),\mathbf z_0)$ is
twice continuously differentiable under standard regularity conditions on the moments.
This justifies the claimed smoothness of $g(\cdot)$.$\square$
\end{proof}

\subsection{Proof of Lemma \ref{t0a}}\label{at0a}
\begin{proof}

Let $\mathbf z_0\in\mathcal W(\mathbf x,\mathbf y)$ be any valid instrument that such that $\mathbf z_0=\mathbf x+k\delta_0 \mathbf r$, and let
$\mathbf e(\delta_0) = \mathbf x - \gamma \mathbf z_0$ denote the associated
first-stage residual in the single-equation, single-IV case ($p=q=1$). 
Under Assumptions
A1--A3, instrument validity implies the orthogonality condition
\[
\mathbb E(\mathbf u \mid \mathbf z_0) = 0.
\]

For the first-stage homoscedasticity moment, the $i$th component of the stacked
moment vector in \eqref{m:1} evaluated at $\delta_0$ is
\[
m_i(\delta_0)
=
\mathbb E\big[(e_i(\delta_0)^2 - \sigma^2)\, z_{0i}\big],
\quad i=1,\dots,n,
\]
so that $\mathbf M(\delta_0) = (m_1(\delta_0),\dots,m_n(\delta_0))'$.
Using the law of iterated expectations and the fact that $z_{0i}$ is
measurable with respect to $\sigma(\mathbf z_0)$,
\[
\begin{aligned}
m_i(\delta_0)
&= \mathbb E\Big[
    \mathbb E\big[(e_i(\delta_0)^2-\sigma^2)\, z_{0i} \mid \mathbf z_0\big]
   \Big] \\
&= \mathbb E\Big[
    z_{0i}\,\mathbb E\big[e_i(\delta_0)^2-\sigma^2 \mid \mathbf z_0\big]
   \Big].
\end{aligned}
\]
Under the homoscedasticity condition
\[
\mathbb E\big[\mathbf e(\delta_0)^{\circ 2} \mid \mathbf z_0\big]
= \sigma^2 \mathbf 1_n \quad\text{a.s.},
\]
we have $\mathbb E[e_i(\delta_0)^2 \mid \mathbf z_0]=\sigma^2$ almost surely for
each $i$, hence
\[
m_i(\delta_0)
= \mathbb E\big[z_{0i}(\sigma^2-\sigma^2)\big]
= 0.
\]
Therefore $m_i(\delta_0)=0$ for all $i$ and thus
\[
\mathbf M(\delta_0) = \mathbf 0.
\]

%

Then,  for the model given by equations \eqref{e2} and \eqref{e2a}, the conditions  ("orthogonality") $\mathbb{E}(\mathbf  u|\mathbf z_0)= {0}$ and ("first-stage homoscedasticity") $\mathbf{M}(\delta_0) = {0}$    hold simultaneously.%
$\square$\end{proof}
\subsection{Proof of Theorem \ref{lm10}}\label{alm10}

\begin{proof}

\noindent\emph{(1) Residual-based characterization (link to Lemma \ref{t0a}).}
Assume again that $\mathbf z_0$ is such that $\mathbb E(\mathbf u\mid\mathbf z_0)=0$.
Instrumenting $\mathbf x$ with an SIV $\mathbf s=\mathbf x+k\,\delta\,\mathbf r$, the first-stage
residual is homoscedastic and given by
\[
\mathbf e = \mathbf x - \gamma\,\mathbf s.
\]
The (scalar) sum of squared residuals can be written as
\[
\mathbf e'\mathbf e
= \|\mathbf x\|^2 + \gamma^2 \|\mathbf s\|^2 - 2\gamma\,\langle\mathbf x,\mathbf s\rangle,
\]
where $\|\cdot\|$ denotes the Euclidean norm and $\langle\cdot,\cdot\rangle$ the inner product.

Treating $\gamma=\gamma(\mathbf s)$ as a smooth function of $\mathbf s$, the gradient of
$\mathbf e'\mathbf e$ with respect to $\mathbf s$ is
\begin{equation}\label{ax51}
\frac{\partial (\mathbf e'\mathbf e)}{\partial \mathbf s}
= \frac{d}{d\mathbf s}\Big(
    \|\mathbf x\|^2
  + \gamma^2\|\mathbf s\|^2
  - 2\gamma\,\langle\mathbf x,\mathbf s\rangle
  \Big)
= 2\gamma \frac{\partial\gamma}{\partial\mathbf s} \|\mathbf s\|^2
  + 2\gamma^2 \mathbf s
  - 2 \frac{\partial\gamma}{\partial\mathbf s} \langle\mathbf x,\mathbf s\rangle
  - 2\gamma\,\mathbf x,
\end{equation}
where $\partial\gamma/\partial\mathbf s$ denotes the gradient of $\gamma$ with respect to $\mathbf s$.

From Lemma~\ref{la1} we have
\[
\gamma(\mathbf s)=\gamma_0 + g(\mathbf x,\mathbf s,\mathbf z_0),
\]
with $\gamma_0=\operatorname{cov}(\mathbf x,\mathbf z_0)/\operatorname{var}(\mathbf z_0)$ and
$f:=g(\mathbf x,\mathbf s,\mathbf z_0)$ capturing the deviation from the true IV. Hence
\begin{equation}\label{ax52}
\frac{\partial\gamma}{\partial\mathbf s}
= \frac{\partial(\gamma_0 + g(\mathbf x,\mathbf s,\mathbf z_0))}{\partial\mathbf s}
= \frac{\partial g(\mathbf x,\mathbf s,\mathbf z_0)}{\partial\mathbf s},
\end{equation}
since $\gamma_0$ does not depend on $\mathbf s$. At the true SIV, $\mathbf s^\star=\mathbf z_0$,
we have $g(\mathbf x,\mathbf s^\star,\mathbf z_0)=0$ and, under the smoothness conditions of
Lemma~\ref{la1}, we take $\partial g(\mathbf x,\mathbf s^\star,\mathbf z_0)/\partial\mathbf s=\mathbf 0$.
Substituting $\mathbf s^\star=\mathbf z_0$, $\gamma=\gamma_0$ and \eqref{ax52} into \eqref{ax51} gives
\begin{equation}\label{ax53}
\left.\frac{\partial (\mathbf e'\mathbf e)}{\partial \mathbf s}\right|_{\mathbf s^\star=\mathbf z_0}
= 2\gamma_0^2 \mathbf z_0 - 2\gamma_0 \mathbf x
= 2\gamma_0(\gamma_0\mathbf z_0 - \mathbf x).
\end{equation}
By construction of the first-stage residuals at the true IV,
\[
\mathbb E(\gamma_0\mathbf z_0 - \mathbf x) = \mathbf 0,
\]
so it follows from \eqref{ax53} that
\[
\mathbb E\left[
 \left.\frac{\partial (\mathbf e'\mathbf e)}{\partial \mathbf s}\right|_{\mathbf s^\star=\mathbf z_0}
\right]
= \mathbf 0.
\]

By Lemma~\ref{t0a}, this condition implies that at the DT condition, the moment vector satisfies
$\mathbf M(\delta)=\mathbf 0$ when $\mathbf s^\star=\mathbf z_0$. Since $\mathbf s^\star:=\mathbf s(\delta_0)=\mathbf z_0$ and $\mathbf z_0$ is a valid IV satisfying the
homoscedasticity condition, the DT moments vanish at $\delta_0$:
\[
\mathbf M(\delta_0)=\mathbf 0.
\]Therefore, $\mathbf M(\delta)=\mathbf 0$  implies that  $\delta=\delta_0$ and 
$\mathbf s^\star=\mathbf z_0$.
Since by definition $\mathbb E(\mathbf u\mid\mathbf z_0)=0$, we conclude that
$\mathbf M(\delta_0)=\mathbf 0$ and $\mathbb E(\mathbf u\mid\mathbf s^\star)=0$
hold simultaneously at $\delta=\delta_0$.

\medskip
\noindent\emph{(2) Local identification via the moment vector.}
Because $\mathbf M(\cdot)$ is $C^1$ in $\delta$ (by A1), we can perform a first-order Taylor
expansion around $\delta_0$. For $\delta$ in a neighborhood $\mathcal N$ of $\delta_0$,
\[
\mathbf M(\delta)
= \mathbf M(\delta_0) 
 + J_{\mathbf M}(\delta_0)\,(\delta-\delta_0) 
 + r(\delta),
\]
where $r(\delta)$ satisfies $\|r(\delta)\|=o\big(|\delta-\delta_0|\big)$ as $\delta\to\delta_0$.
Since $\mathbf M(\delta_0)=\mathbf 0$, this simplifies to
\begin{equation}\label{eq:taylor}
\mathbf M(\delta)
= J_{\mathbf M}(\delta_0)\,(\delta-\delta_0) 
  + o\big(|\delta-\delta_0|\big).
\end{equation}


Since $J_{\mathbf M}(\delta_0) \neq \mathbf 0$ is a nonzero column vector, 
the scalar $\|J_{\mathbf M}(\delta_0)\|^2 > 0$. Taking $v' = J_{\mathbf M}(\delta_0)'$,
\[
v'\mathbf M(\delta) = J_{\mathbf M}(\delta_0)'J_{\mathbf M}(\delta_0)(\delta-\delta_0) + o(|\delta-\delta_0|)
= \|J_{\mathbf M}(\delta_0)\|^2(\delta-\delta_0) + o(|\delta-\delta_0|).
\]
Since $\|J_{\mathbf M}(\delta_0)\|^2 > 0$, the right-hand side is nonzero for 
$\delta$ in a punctured neighborhood 
$\mathcal N^*(\delta_0) := \{\delta : 0 < |\delta - \delta_0| < \epsilon\}$ 
of $\delta_0$. Therefore $\mathbf M(\delta) \neq \mathbf 0$ for all $\delta \in \mathcal N^*(\delta_0)$.

\medskip
\noindent\emph{(3) Local minimum of the GMM objective.}
Consider the GMM objective
\[
J(\delta)=\mathbf M(\delta)'\mathbf W\mathbf M(\delta),
\]
where $\mathbf W$ is a positive definite weighting matrix. Substituting \eqref{eq:taylor} into $J(\delta)$,
we obtain
\[
J(\delta)
= \big(J_{\mathbf M}(\delta_0)(\delta-\delta_0) + o(|\delta-\delta_0|)\big)' 
  \mathbf W
  \big(J_{\mathbf M}(\delta_0)(\delta-\delta_0) + o(|\delta-\delta_0|)\big).
\]
Expanding the quadratic form yields
\[
J(\delta)
= (\delta-\delta_0)^2\,J_{\mathbf M}(\delta_0)'\mathbf WJ_{\mathbf M}(\delta_0)
  + o\big((\delta-\delta_0)^2\big).
\]
Since $\mathbf W$ is positive definite and $J_{\mathbf M}(\delta_0)\neq\mathbf 0$, the leading
coefficient $J_{\mathbf M}(\delta_0)'\mathbf WJ_{\mathbf M}(\delta_0)$ is strictly positive.
Thus, for $\delta$ sufficiently close to $\delta_0$ and $\delta\neq\delta_0$,
\[
J(\delta) > 0 = J(\delta_0),
\]
so $\delta_0$ is a strict local minimizer of $J(\delta)$.

\medskip

Combining parts (1)--(3) yields the stated result.$\square$
\end{proof}

{\color{blue}

}
\subsection{Proof of Corollary \ref{comp}}\label{acomp}
\begin{proof}We divide the proof into four parts: (i) we show the objective function convergence, (ii) define the population minimizer,  (iii) show the consistency of the sample minimizer, (iv) and then show  the consistency of the estimated SIV.

\emph{(i) Objective function convergence.}
Define the sample objective function
\[
J_n(\delta) := \hat{\mathbf{M}}_n(\delta)'\mathbf{W}\hat{\mathbf{M}}_n(\delta)
\]
and the population objective function
\[
J(\delta) := \mathbf{M}(\delta)'\mathbf{W}\mathbf{M}(\delta).
\]
By uniform law of large numbers, under regularity conditions on the moment functions, we have
\[
\sup_{\delta\in(0,\bar{\delta})} |J_n(\delta) - J(\delta)| \xrightarrow{p} 0
\]
as $n\to\infty$. This ensures uniform convergence of the sample criterion to its population counterpart.

\emph{(ii) Population minimizer.}
By Theorem~\ref{lm10}, $\mathbf{M}(\delta_0) = \mathbf{0}$, which implies $J(\delta_0) = 0$. For $\delta \neq \delta_0$ in a neighborhood of $\delta_0$, the continuous differentiability of $\mathbf{M}(\delta)$ and the condition $J_{\mathbf{M}}(\delta_0) \neq \mathbf{0}$ ensure that
\[
\mathbf{M}(\delta) = \mathbf{M}(\delta_0) + J_{\mathbf{M}}(\delta_0)(\delta - \delta_0) + o(|\delta - \delta_0|) = J_{\mathbf{M}}(\delta_0)(\delta - \delta_0) + o(|\delta - \delta_0|).
\]
Since $\mathbf{W}$ is positive definite and $J_{\mathbf{M}}(\delta_0) \neq \mathbf{0}$, we have
\[
J(\delta) = (\delta - \delta_0)^2 J_{\mathbf{M}}(\delta_0)'\mathbf{W}J_{\mathbf{M}}(\delta_0) + o((\delta - \delta_0)^2) > 0
\]
for $\delta \neq \delta_0$ sufficiently close to $\delta_0$. Thus, $\delta_0$ is a strict local minimizer of $J(\delta)$.

\emph{(iii) Consistency via convergence of minimizers.}
By the uniform convergence established  in (i) and the identification of $\delta_0$ as a strict local minimizer in (ii), standard results on extremum estimators (e.g., Theorem 5.7 in \citet{vanderVaart1998}) imply
\[
\hat{\delta}_{0n}  \xrightarrow{p} \delta_0.
\]

\emph{(iv) Consistency of the estimated SIV.}
Since $\hat{\mathbf{s}}^\star = \mathbf{x} + k\hat{\delta}_{0n}\mathbf{r}$ and the mapping $\delta \mapsto \mathbf{x} + k\delta\mathbf{r}$ is continuous, by the continuous mapping theorem,
\[
\hat{\mathbf{s}}^\star = \mathbf{x} + k\hat{\delta}_{0n} \mathbf{r} \xrightarrow{p} \mathbf{x} + k\delta_0\mathbf{r} = \mathbf{s}^\star.
\]
Since
$\plim\,\hat{\mathbf{s}}^\star = \mathbf{s}^\star$, it follows that
$\plim_{n\to\infty}\mathbb E(\mathbf{u}\mid\hat{\mathbf{s}}^\star)=0$,
confirming asymptotic validity. $\square$
\end{proof}

\subsection{Proof of Corollary \ref{tx1}}\label{atx1}
\begin{proof}

According to Theorem \ref{lm10}, the DT condition, $\delta_0 = \operatorname*{arg\,}_{\delta}\{ \mathbf{M}(\delta)=\mathbf 0\}$,  determines a valid SIV,  ${\mathbf s^\star}={\mathbf x}+k\delta_0 {\mathbf r}$ with $k=(-1)sign[\operatorname{cov}(\mathbf x,\mathbf u)]$ such that $\mathbb{E}({\mathbf u'\mathbf s^\star})=0$. Then, $\beta$, a parameter of the model in  in \eqref{e2}, is identified by an IV estimator: ${\beta}_{IV}={(\mathbf x'\mathbf s^\star)^{-1} \mathbf x'\mathbf y}$.$\square$
\end {proof}

\subsection{Proof of Corollary \ref{c2a}}\label{ac2a}

\begin{proof}
We work in the single-equation case with $p=q=1$ and determine the sign of $\operatorname{corr}(\mathbf{x},\mathbf{u})$ by testing two sign hypotheses and checking which one admits a valid synthetic instrument.

\medskip
\noindent\textit{1. DT condition $\Leftrightarrow$ covariance condition.}

By Theorem~\ref{lm10}, if there exists a valid instrument $\mathbf{z}_0$ with $\mathbb{E}(\mathbf{u}\mid \mathbf{z}_0)=\mathbf{0}$ and $\mathbb{E}(\mathbf{e}\mathbf{e}'\mid \mathbf{z}_0)=\sigma^2\mathbf{I}_n$, then at $\delta_0$ with $\mathbf{s}(\delta_0)=\mathbf{z}_0$ the DT moment condition is
\[
\mathbf{M}(\delta_0)=\mathbb{E}[(\mathbf{e}^{\circ 2}-\sigma^2)\mathbf{z}_0]=\mathbf{0},
\]
where $\mathbf{e}=\mathbf{x}-\gamma\mathbf{s}$ and $\mathbf{s}=\mathbf{z}_0$. Rearranging gives
\[
\mathbb{E}[\mathbf{e}^{\circ 2}\mathbf{z}_0] = \sigma^2\mathbb{E}[\mathbf{z}_0].
\]
Under homoscedasticity, $\mathbb{E}[\mathbf{e}^{\circ 2}] = \sigma^2$, so
\[
\operatorname{cov}(\mathbf{e}^{\circ 2},\mathbf{z}_0)
= \mathbb{E}[\mathbf{e}^{\circ 2}\mathbf{z}_0] - \mathbb{E}[\mathbf{e}^{\circ 2}]\mathbb{E}[\mathbf{z}_0]
= \sigma^2\mathbb{E}[\mathbf{z}_0]-\sigma^2\mathbb{E}[\mathbf{z}_0]
= \mathbf{0}.
\]
Thus, $\mathbf{M}(\delta_0)=\mathbf{0}$ iff $\operatorname{cov}(\mathbf{e}^{\circ 2},\mathbf{s}^\star)=\mathbf{0}$ where $\mathbf{s}^\star=\mathbf{z}_0=\mathbf{x}+k\delta_0\mathbf{r}$.

\medskip
\noindent\textit{2. Sign hypotheses and candidate instruments.}

We consider
\[
(+) : \operatorname{corr}(\mathbf{x},\mathbf{u})>0, 
\qquad
(-) : \operatorname{corr}(\mathbf{x},\mathbf{u})<0.
\]
For each sign, we construct $\mathbf{s}=\mathbf{x}+k\delta\mathbf{r}$ with $\delta>0$, choosing the orientation of $\mathbf{r}$ via $k$ so that increasing $\delta$ moves $\mathbf{s}$ in the direction that (under the maintained sign) reduces endogeneity. Under $(-)$ we assume $k=1$, whereas under $(+)$, $k=-1$.

\medskip
\noindent\textit{3. Testing hypothesis $(+)$.}

Under $(+)$, define $\mathbf{s}_{(+)}(\delta)=\mathbf{x}-\delta\mathbf{r}$ and the corresponding first-stage residual
\[
\mathbf{e}_{(+)}(\delta)=\mathbf{x}-\gamma_{(+)}(\delta)\mathbf{s}_{(+)}(\delta),
\qquad
\gamma_{(+)}(\delta)=\frac{\langle \mathbf{x},\mathbf{s}_{(+)}(\delta)\rangle}{\|\mathbf{s}_{(+)}(\delta)\|^2}.
\]
We test whether there exists $\delta_{0,+}>0$ such that
\begin{equation}\label{eq:cov-plus-condensed}
\operatorname{cov}\!\big(\mathbf{e}_{(+)}^{\circ 2}(\delta_{0,+}),\mathbf{s}_{(+)}(\delta_{0,+})\big)=0.
\end{equation}
If so, Step 1 and Theorem~\ref{lm10} imply that $\mathbf{s}_{(+)}^\star:=\mathbf{s}_{(+)}(\delta_{0,+})$ is a valid IV with $\mathbb{E}(\mathbf{u}\mid \mathbf{s}_{(+)}^\star)=\mathbf{0}$. If moreover $J_{\mathbf{M}}(\delta_{0,+})\neq 0$, then $\delta_{0,+}$ is locally unique.

\medskip
\noindent\textit{4. Testing hypothesis $(-)$.}

Under $(-)$, set $\mathbf{s}_{(-)}(\delta)=\mathbf{x}+\delta\mathbf{r}$ and define
\[
\mathbf{e}_{(-)}(\delta)=\mathbf{x}-\gamma_{(-)}(\delta)\mathbf{s}_{(-)}(\delta).
\]
We test for $\delta_{0,-}>0$ such that
\begin{equation}\label{eq:cov-minus-condensed}
\operatorname{cov}\!\big(\mathbf{e}_{(-)}^{\circ 2}(\delta_{0,-}),\mathbf{s}_{(-)}(\delta_{0,-})\big)=0.
\end{equation}

If~\eqref{eq:cov-plus-condensed} holds for some $\delta_{0,+}>0$ but~\eqref{eq:cov-minus-condensed} fails for all $\delta>0$, then by Theorem~\ref{lm10} an IV exists only under $(+)$, so the true sign is $\operatorname{corr}(\mathbf{x},\mathbf{u})>0$. By symmetry, if only $(-)$ yields a valid IV, then $\operatorname{corr}(\mathbf{x},\mathbf{u})<0$.

\medskip
\noindent\textit{5. Remaining cases.}

If both hypotheses fail (no $\delta>0$ solves either covariance condition), then no valid synthetic instrument lies in $\mathcal{W}(\mathbf{x},\mathbf{y})$. Under Assumptions A1--A3, this implies $\operatorname{corr}(\mathbf{x},\mathbf{u})=0$, i.e., no endogeneity.

If both hypotheses succeed (solutions $\delta_{0,+}>0$ and $\delta_{0,-}>0$ exist), this procedure does not sign-identify $\operatorname{corr}(\mathbf{x},\mathbf{u})$. Under the maintained assumption $\operatorname{corr}(\mathbf{x},\mathbf{u})\neq 0$ and Assumptions A1--A3 (which ensure identification), this outcome is nongeneric; if it arises, further identifying information is needed.$\square$
\end{proof}

\subsection{Proof of Theorem \ref{l:x1}}\label{a:x1}
\begin{proof}

Let \[
\mathbf H(\delta) \;:=\; \mathbb E\!\big[\mathbf e(\delta)\,\mathbf e(\delta)'\,\big|\,\mathbf s\big] \;\in\; \mathbb R^{n\times n},
\] express as $\mathbf H = \mathbf H(\zeta)$ be a $n \times n$ matrix depending on scalar parameter $\zeta$ defined in equation \eqref{x1s}, and denote $\mathbf H'(\zeta) := \frac{\partial\mathbf H}{\partial\zeta}$. Recall that $D(\delta)
\;=\; \big\|\mathbf H(\delta)-\mathbf I_n\big\|_F^2
\;=\; \operatorname{tr}\!\Big(\big[\mathbf H(\delta)-\mathbf I_n\big]^2\Big).$
We expand it as and use $\zeta$ as the parameter depending on deeper parameter $\delta$: $D(\zeta) = \operatorname{tr}[\mathbf H^2] - 2\,\operatorname{tr}[\mathbf H] + n$ with $\mathbf H =\mathbf  H(\zeta)$ a $n\times n$ matrix, 
using $\tfrac{\partial}{\partial\zeta}\operatorname{tr}[\mathbf H] = \operatorname{tr}[\mathbf H']$ and 
$\tfrac{\partial}{\partial\zeta}\operatorname{tr}[\mathbf H^2] = 2\operatorname{tr}[\mathbf H'\mathbf H]$, we obtain
\begin{equation}
\frac{\partial D}{\partial\zeta}
= 2\,\operatorname{tr}[\mathbf H'(\mathbf H - \mathbf I_n)].
\label{eq:dD}
\end{equation}
\begin{equation}
{\frac{\partial D}{\partial\zeta} = 0 \iff \text{tr}[\mathbf d'({\mathbf H - \mathbf I_n)}] = 0.}
\end{equation}

Expanding:
\begin{equation}
\text{tr}[\mathbf d' \mathbf H] - \text{tr}[\mathbf d' \mathbf I_n] = 0 \implies \text{tr}[\mathbf d' \mathbf H] = \text{tr}[\mathbf d].
\end{equation}
However, $\mathbf d \perp H$ requires:
\begin{equation}
\text{tr}[\mathbf d' H] = 0
\end{equation}
These are equivalent \emph{only if} $\text{tr}[\mathbf d] = 0$.

Assuming the consistent FGLS specification,  for $\mathbf H(\zeta) = \sigma^2 I_n + \zeta\mathbf  d + \alpha \mathbf z_0$, we have $\mathbf H'(\zeta) = \mathbf d$, so
\[
\frac{\partial D}{\partial\zeta}
= 2\!\left[(\sigma^2-1)\operatorname{tr}[\mathbf d]
+ \zeta\,\operatorname{tr}[\mathbf d^2]
+ \alpha\,\operatorname{tr}(\mathbf d \mathbf z_0)\right].
\]
Setting $\tfrac{\partial{\mathbf D}}{\partial\zeta} = 0$ and noting $\tfrac{\mathbf \partial^2D}{\partial\zeta^2} = 2\operatorname{tr}(\mathbf d^2) > 0$ yields the global minimum
\begin{equation}
\zeta^*
= \frac{(1-\sigma^2)\operatorname{tr}[\mathbf d] - \alpha\operatorname{tr}[\mathbf d\mathbf z_0]}
{\operatorname{tr}[\mathbf d^2]}.
\label{eq:zeta_star}
\end{equation}

From~\eqref{eq:zeta_star}, we conclude that
\begin{equation}
{
\zeta^* = 0
\;\iff\;
(1-\sigma^2)\operatorname{tr}[\mathbf d] = \alpha\operatorname{tr}[\mathbf dz_0].
}
\end{equation}
 \text{If $\alpha = 0$},
\(
\zeta^* = 0 \implies \operatorname{tr}[\mathbf d] = 0.
\) This case holds when we have a homoscedastic case. However, we have a heteroscedastic case ($\alpha\neq0$), then
 \text{if $\mathbf d \perp \mathbf z_0$ (i.e., $\operatorname{tr}[\mathbf d\mathbf z_0] = 0$):}
\(
\zeta^* = 0 \implies \operatorname{tr}[\mathbf d] = 0.
\) That is, when  any deviations from the true IV $\mathbf z_0$ satisfy $\text{tr}[\mathbf d] = 0$, $\implies \mathbf d \perp \mathbf H$. Therefore, a SIV ${\mathbf s-\mathbf d}|_{\frac{\partial D}{\mathbf \partial\zeta} = 0} \xrightarrow[]{p} {\mathbf s^\star=\mathbf z_0}$.
$\square$
\end{proof}

\subsection{Proof of Lemma \ref{lemma:uc}}\label{alemma:uc}
\begin{proof} For any $\delta\in\mathcal{D}$ write
\[
\big| \widehat D_n(\delta) - D(\delta)\big|
= \big| \|{\mathbf {\widehat H_n}(\delta)-\mathbf I}\|_F^2 - \|{\mathbf H(\delta)-\mathbf I}\|_F^2 \big|.
\]
Use the algebraic identity $a^2-b^2=(a-b)(a+b)$ with $a=\|{\mathbf {\widehat H_n}}(\delta)-{\mathbf I}\|_F$ and $b=\|{\mathbf H(\delta)-\mathbf I}\|_F$ to obtain
\[
\big| \widehat D_n(\delta) - D(\delta)\big|
\le \big\|{\mathbf {\widehat H_n}(\delta)-\mathbf H}(\delta)\big\|_F\;\big( \|{\mathbf {\widehat H_n}}(\delta)-\mathbf I\|_F + \|{\mathbf H(\delta)-\mathbf I}\|_F \big).
\]
By Assumption A6, $\sup_{\delta}\|{\mathbf {\widehat H_n}(\delta)-\mathbf H}(\delta)\|_F \xrightarrow{p}0$. Also continuity of $\mathbf H(\delta)$ on compact $\mathcal{D}$ implies $\sup_{\delta}\|{\mathbf H(\delta)-\mathbf I})\|_F<\infty$. Moreover
\[
\sup_{\delta}\|{\mathbf {\widehat H_n}(\delta)-\mathbf I}\|_F \le \sup_{\delta}\|{\mathbf {\widehat H_n}(\delta)-H}(\delta)\|_F + \sup_{\delta}\|{\mathbf H(\delta)-\mathbf I}\|_F,
\]
so the bracket term is stochastically bounded. Hence, the right-hand side is $o_p(1)$ uniformly in $\delta$, proving the lemma. \(\square\)\end{proof}

\subsection{Proof of Proposition \ref{prop:ident_consis}}\label{aprop:ident_consis}
\begin{proof}
We divide the proof into three parts: (i) continuity and attainment of the minimum, (ii) population identification, and (iii) consistency of the sample minimizer.

\emph{(i) Continuity and attainment.}

In the single–equation case $p=1$, we work with the stacked first–stage residual
vector $\mathbf e(\delta)\in\mathbb R^n$ and its conditional second–moment
matrix
\[
\mathbf H(\delta)
:= \mathbb E\big[\mathbf e(\delta)\mathbf e(\delta)' \mid \mathbf s\big]
\in\mathbb R^{n\times n}.
\]
By Assumption A5, the map $\delta\mapsto \mathbf H(\delta)$ is continuous on the
compact set $\mathcal D$, with $\mathbf H(\delta)$ symmetric for each
$\delta\in\mathcal D$.
Define the discrepancy
\[
D(\delta)
:= \big\|\mathbf H(\delta)-\mathbf I_n\big\|_F^2
= \operatorname{tr}\Big( \big[\mathbf H(\delta)-\mathbf I_n\big]^2 \Big).
\]
The mapping
\[
\mathbf M \;\mapsto\; \|\mathbf M-\mathbf I_n\|_F^2
\]
from the space of $n\times n$ real matrices to $\mathbb R$ is continuous (as a polynomial) in the entries of $\mathbf M$. Hence, by composition with the
continuous map $\delta\mapsto \mathbf H(\delta)$, the function
\[
\delta \;\mapsto\; D(\delta)
= \big\|\mathbf H(\delta)-\mathbf I_n\big\|_F^2
\]
is continuous on $\mathcal D$. Since $\mathcal D$ is compact, $D$ attains its
minimum on $\mathcal D$.


\emph{(ii) Population identification (uniqueness).}
By theorem \ref{l:x1},  
 the minimum value of $D$  is attained uniquely at $\delta_0$. This proves population identification.

\emph{(iii) Consistency of the sample minimizer.}
We will apply the standard argmin-consistency theorem (Newey--McFadden style). 
\medskip\noindent From Lemma \ref{lemma:uc}, we have $\widehat D_n \overset{p}{\to} D$ (uniform convergence in probability) on the compact set $\mathcal{D}$. By (ii) above, $D$ is continuous and has a unique minimizer at $\delta_0$. The argmin-consistency theorem (e.g.\ Theorem 2.1 in \citet{newey1994large}) states that if $\widehat D_n$ converges uniformly in probability to a continuous function $D$ that has a unique minimizer, then any sequence of measurable approximate minimizers $\hat\delta_n$ converges in probability to that unique minimizer. Applying that theorem yields $\hat\delta_n\xrightarrow{p}\delta_0$.$\square$
\end{proof}

\subsection{Proof of Corollary \ref{tx2}}\label{atx2}
\begin{proof}
{According to Theorem \ref{l:x1}, the robust DT condition, ${\delta_0}=\operatorname*{arg\,min}_{\delta}( { D})$, determines a valid SIV,  ${\mathbf s^\star}={\mathbf x}+k\delta_0 {\mathbf r}$ with $k=(-1)sign[\operatorname{cov}(\mathbf x,\mathbf u)]$ such that $\mathbb{E}({\mathbf u'\mathbf s^\star})=0$. Then, $\beta$, a parameter of the model in  in \eqref{e2}, is identified by an IV estimator: ${\beta}_{IV}={(\mathbf x'\mathbf s^\star)^{-1} \mathbf x'\mathbf y}$}.
\end{proof}

\subsection{Proof of Lemma \ref{lx3}}\label{alx3}
{According to  Theorem \ref{l:x1}, for the difference between the conditional variance of the first-stage error terms, ${D}(\delta)$ determined for all $\delta \in(0,\bar{\delta})$,  
at point ${\delta_0}=\operatorname*{arg\,min}_{\delta}( { D})$, $\mathbb{E}({\mathbf  u\mid {\mathbf s^\star=\mathbf z_0}})=0$ holds.  
According to Lemma \ref{lemma:uc}, ${ \widehat D_n }{\overset{p}{\to} }{ D}$}. Therefore, \mbox{$\hat{\delta}_0=\operatorname*{arg\,min}_{\delta}( {\mathbf D}_E= \{\widehat{D}_n(\delta):\delta\in(0,\bar{\delta})\}$}  identifies the SIV ${\mathbf {\hat s}^\star}={\mathbf x}+k\hat{\delta}_0 {\mathbf r}$  such that  ${\mathbf {\hat s}^\star\xrightarrow[]{p}\mathbf z_0}$ holds. Since by construction $\mathbb{E}({\mathbf u\mid \mathbf z_0})=0$, $  \hat{\mathbf s}^\star\xrightarrow[]{p}{\mathbf z_0}$ implies that   \mbox{$\plim \limits_{n\rightarrow \infty}\mathbb E({\mathbf u\mid  \mathbf {\hat{s}}^\star})=0$.}$\square$

\subsection{Proof of Lemma \ref{l2a}}\label{al3a}
\begin{proof}
 Let us denote by ${\mathbf V_0=\mathbf V|\mathbf s^\star}$  the matrix of exogenous variables that includes the true IV $\mathbf s^\star$ determined as an SIV that satisfies the DT condition, and denote by $\mathbf X=\mathbf V|\mathbf x$ the matrix of the regressors.  A standard assumption for the IV estimator to be consistent is $\plim\limits_{n\rightarrow \infty}n^{-1}{\mathbf V_0' \mathbf u= 0}$. That is, the error terms are asymptotically uncorrelated with the instruments.
We can express the SIV estimator as
\begin{eqnarray}\label{a9a}\mathbf b_{SIV} = { (\mathbf V'_0\mathbf x)^{-1}\mathbf V'_0 \mathbf X\beta_0 + (\mathbf V'_0\mathbf  X)^{-1}\mathbf V'_0 \mathbf u}\\\nonumber =\beta_0 + (n^{-1}{\mathbf V'_0 \mathbf X)^{-1}} n^{-1}{\mathbf V'_0 \mathbf u}.\end{eqnarray}
Since the  $\plim \limits_{n\rightarrow \infty}n^{-1}{ (\mathbf V'_0 \mathbf X)^{-1}}${ is deterministic and nonsingular by assumption},  then $\mathbf b_{SIV}$
satisfies the first asymptotic identification condition because 
$\plim\limits_{n\rightarrow \infty}n^{-1}{\mathbf V'_0 \mathbf u= 0}$ holds due to Lemma \ref{lx3}. 

Next, let us consider the second condition for $\beta\neq\beta_0$.
For the true DGP, we have \begin{equation}\label{a92}{ \mathbf y=\mathbf X\beta_0+\mathbf u}.\end{equation} We transform \eqref{a92} to
\begin{equation}\label{a93} {\mathbf y-\beta \mathbf x=\mathbf X\beta_0+\mathbf u-\mathbf X\beta =\mathbf u+\mathbf X(\beta_0-\beta)}.\end{equation}
Using\eqref{a93}, we write $${\mathbf \alpha(\beta) }=\plim \limits_{n\rightarrow \infty}n^{-1}{\mathbf V'_0 (\mathbf y-\beta \mathbf x)}$$
$$=\plim\limits_{n\rightarrow \infty}n^{-1}{\mathbf V'_0 (\mathbf u+\mathbf X(\beta_0-\beta))}. $$
Since $\plim\limits_{n\rightarrow \infty}n^{-1}{\mathbf V'_0\mathbf  u= 0}$, we have
$${\mathbf \alpha(\beta) }=\plim \limits_{n\rightarrow \infty}n^{-1}{\mathbf V'_0 \mathbf X(\beta_0-\beta)}. $$
It is known  that $\plim \limits_{n\rightarrow \infty}n^{-1}{ (\mathbf V'_0 \mathbf X)^{-1}}${can be assumed as deterministic and nonsingular}, thus, the  probability limit ${\mathbf \alpha(\beta) }\neq 0$ as soon as $\mathbf\beta\neq \beta_0$.The second asymptotic identification condition is satisfied; therefore, the SIV estimator is consistent.
$\square$\end{proof}

%
%
%

\subsection{Step-by-step guide to applying the SIV method}\label{sec:algorithm}
\begin{enumerate}
\item Define $\mathbf{y}$ and $\mathbf{x}$ as residuals from a projection onto the space spanned by $\tilde{\mathbf{V}}$, that is:
$$\mathbf{y} = (\mathbf{I} - \mathbf{P}{\tilde{\mathbf{V}}}) \tilde{\mathbf{y}}\,\text{ and } 
\mathbf{x} = (\mathbf{I} - \mathbf{P}{\tilde{\mathbf{V}}}) \tilde{\mathbf{x}},$$
where $\tilde{\mathbf{x}}$ and $\tilde{\mathbf{y}}$ denote the original vectors for the outcome variable and the endogenous regressor, $\mathbf{I}$ is the identity matrix, and $\mathbf{P}_{\tilde{\mathbf{V}}}$ is the projection matrix onto the vector space spanned by the matrix of exogenous regressors $\tilde{\mathbf{V}}$.
\item Choose a vector $\mathbf{r} \perp \mathbf{x}$ in plane $\mathcal{W}$ spanned by $\mathbf x$ and $\mathbf y$. In practice, one can use the residual of the regression $\mathbf x=\beta y+\epsilon$
\item  Follow Corollary \ref{c2a} and determine the true sign of $\operatorname{cov}(\mathbf{x, u})$ based on the results obtained from both sign assumptions. Assume the sign of $\operatorname{cov}({\mathbf x,u})$. 
\item Assume the starting value for the scalar parameter $\delta$ to be a small positive number, e.g., $0.001$.
\item \label{i11} Construct $\mathbf{s} = \mathbf{x} + k\delta \mathbf{r}$, where $k = (-1) \cdot \text{sign}({\operatorname{cov}}(\mathbf{x}, \mathbf{u}))$.
\item Obtain first-stage residuals $\hat{\mathbf{e}}_i$ and estimate
$\hat{\sigma}^2 = \tfrac{1}{np}\sum_i\operatorname{tr}(\hat{\mathbf{e}}_i\hat{\mathbf{e}}_i')$.
\item Form the moment vector
$\hat{m}_i = \operatorname{vec}(\hat{\mathbf{e}}_i\hat{\mathbf{e}}_i' - \hat{\sigma}^2\mathbf{I}) \otimes \mathbf{s}_i$. In our case, it is just $\hat{m}_i(\delta)=(\hat e_i^2-\hat\sigma^2)\,s_i$.
\item Estimate the covariance
$\hat{\mathbf{S}} = \tfrac{1}{n}\sum_i (\hat{m}_i - \bar{m})(\hat{m}_i - \bar{m})'$.
\item Compute the test statistic
\[
J = n \bar{m}' \hat{\mathbf{S}}^{-1} \bar{m},
\]
which under the null is asymptotically $\chi^2_d$ distributed.
The optimal $\delta_0$ is obtained as
\[
\hat{\delta}_0 = \arg\min_{\delta\in\Delta} J_n(\delta).
\]

\item \label{i13} Compute $\mathbf{s}_0 = \mathbf{x} + k \delta_0 \mathbf{r}$.

    \item Use $\mathbf{s}_0$ found for the true sign as the instrumental variable in the two-stage least squares (2SLS) estimation procedure to obtain consistent estimates of the causal effect of $\mathbf{x}$ on $\mathbf{y}$.
\end{enumerate}

\textbf{Bootstrap inference.} Because $\delta_0$ is estimated in a first step, the asymptotic distribution of $\hat{\beta}_{\text{SIV}}$ may differ from standard IV in finite samples. We recommend bootstrap inference:
\begin{enumerate}
\item Draw $B$ bootstrap samples from the data.
\item For each bootstrap sample $b = 1, \ldots, B$:
\begin{itemize}
\item Estimate $\hat{\delta}_0^{(b)}$ via the DT condition.
\item Construct $\hat{\mathbf s}^{\star(b)} = \mathbf x^{(b)} + k\hat{\delta}_0^{(b)} \mathbf r^{(b)}$.
\item Estimate $\hat{\beta}_{\text{SIV}}^{(b)}$ using $\hat{\mathbf s}^{\star(b)}$ as the instrument.
\end{itemize}
\item Use the bootstrap distribution of $\{\hat{\beta}_{\text{SIV}}^{(b)}\}_{b=1}^B$ for inference (standard errors, confidence intervals, hypothesis tests).
\end{enumerate}
This bootstrap procedure accounts for the sampling variability in $\hat{\delta}_0$ and typically provides more accurate finite-sample inference than asymptotic approximations.
The following procedures are applied to either a dataset or a bootstrap sample. In the case of bootstrapping, it is essential to save the values of \(\hat{\beta}\) and \(\delta_0\) obtained for each sample. The means of the sample \(\hat{\beta}\) can be reported directly as the SIV estimates. Additionally, one can use the mean of the sample \(\delta_0\) to determine the SIV using this estimate. Subsequently, the IV estimation method can be applied using the obtained SIV as the instrumental variable.%
\subsection{Step-by-step guide to the heteroscedasticity robust SIV method}
We assume that the true sign of $cor({\mathbf x, \mathbf u})$ is determined using the simple approach above.
\begin{enumerate}

\item Steps 1-5 of the above procedure 
\item  Compute the residuals of the first-stage regression: $\mathbf{e} = \mathbf{x} - \gamma_{OLS} \mathbf{s}$ and $\mathbf{e}_g = \mathbf{x} - \gamma_{FGLS} \mathbf{s}$. 
\item Estimate predicted values of  regressions $\bf e^2=s+\epsilon$ and $\bf e_g^2=s+\epsilon_g$.
\item \textbf{Parametric approach:}  Compute $X^{2}=\dfrac{\textrm{SSR}/2}{(\textrm{SSE}/n)^{2}}$ for OLS and $ X_g^{2}=\dfrac{\textrm{SSR}_g/2}{(\textrm{SSE}_g/n)^{2}}$, for FGLS case, where $\textrm{SSR}=\sum_{i=0}^n(\hat{\hat{e}}_i-\bar{\hat{e}})^2$ and $\textrm{SSR}_g=\sum_{i=0}^n(\hat{\hat{e}}_{gi}-\bar{\hat{e}}_g)^2.$
\item \label{i2} Determine ${ D(\delta)}=P({\chi^2(1)<X^{2}(\delta)})/P({\chi^2(1)<X^{2}_{g}(\delta)}),$
 for all $\delta \in (0, \bar{\delta})$, and construct  the  locus:  ${ \mathbf D_E}=\{{D(\delta=0)..., \, D(\delta=\bar{\delta})}\}$.

\item \textbf{Non-parametric approach:} For each $\delta$, let $F_n(\delta)$ and $G_n(\delta)$ denote the empirical CDFs of $\{\hat{e}_i^2(\delta)\}$ and $\{\hat{e}_{gi}^2(\delta)\}$, respectively. Define the two-sample Anderson-Darling statistic:
$$A_{n,m}^2(\delta) = \frac{nm}{(n+m)^2} \sum_{k=1}^{n+m} \frac{[F_n(x_k) - G_m(x_k)]^2}{H_{n+m}(x_k)[1 - H_{n+m}(x_k)]},$$
where $H_{n+m}$ is the combined empirical CDF.

\item Construct the locus $\mathbf{D_E} = \{ A_{n,m}^2(\delta): \delta \in (0,\bar{\delta})\}$.

    \item Determine \mbox{$\delta_0=\operatorname*{arg\,min}_{\mathbf \delta}( \mathbf{ D_E})$}.
\item Compute $\mathbf{s}_0 = \mathbf{x} + k \delta_0 \mathbf{r}$.
    \item Use $\mathbf{s}_0$ as the instrumental variable in the two-stage least squares (2SLS) estimation procedure to obtain consistent estimates of the causal effect of $\mathbf{x}$ on $\mathbf{y}$.
\end{enumerate}

\subsection{Multiple endogenous variables}\label{app:multiple_endog}
If there are multiple endogenous variables, the SIV method can be applied with some adjustments.
We begin by examining a two-variable model based on \citet{ang1}:
\begin{eqnarray}\label{me1} {\mathbf y }= {\mathbf x_1}\beta_1 +{\mathbf  x_2}\beta_2 + {\mathbf u },\\
{\mathbf x_1} = {\mathbf Z}\Pi_1+{\mathbf v_1},\nonumber\\
{\mathbf x_2 }= {\mathbf Z}\Pi_2 + {\mathbf v_2},\nonumber \end{eqnarray}
where $\mathbf y$, $\mathbf x_1$, $\mathbf x_2$, $\mathbf u$, $\mathbf v_1$ and $\mathbf v_2$ are $n \times 1$ vectors, with $n$ the number
of observations. $\mathbf Z$ is an $n \times k_z$ matrix of instruments, with $k_z \geq 2$
($\Pi_1$ and $\Pi_2$ are $k_z \times 1$ vectors).

Here, as in the single endogenous variable case, we recall that any model with
additional matrix ${\mathbf V}\in\mathcal{H}$ of  predetermined or exogenous regressors, including a vector of ones, can be reduced to this form
by defining $\mathbf y$,
$\mathbf x_1$, $\mathbf x_2$, and $\mathbf Z$ as residuals from the orthogonal projection onto the closed linear subspace spanned by ${\mathbf V}$. That is,
${\mathbf y}=(\mathbf I-P_{{\mathbf V}} ) \tilde{\mathbf y}$,
${\mathbf x_1}=(\mathbf I-P_{{\mathbf V}}) \tilde{\mathbf x}$, ${\mathbf x_2}=(\mathbf I-P_{{\mathbf V}}) \tilde{\mathbf x}$, ${\mathbf Z}=(\mathbf I-P_{{V}} ) \tilde{\mathbf Z}$, where $\mathbf P_{{V}}$ be the orthogonal projection matrix onto $\mathrm{span}({\mathbf V})$,  and $ \mathbf I$ is the identity matrix on $\mathcal{H}$; ${\mathbf y}$,  $\tilde{\mathbf x}_1$,  $\tilde{\mathbf x}_2$,  and $\tilde{\mathbf Z}$ denotes the original vectors in this extended case. %
Using the Frisch-Waugh-Lowell (FWL) Theorem, $\beta_1$ in \eqref{me1} can be estimated using the follwoing regression instead
\begin{equation}\label{me2}
{\mathbf y=\mathbf M_2 \mathbf x_1}\beta_1+{\mathbf v},
\end{equation}
where $\mathbf M_2=\mathbf I-\mathbf P_2$ is the ortogonal projection, with $\mathbf I$ being the identity vector and $\mathbf P_2=\mathbf x_2(\mathbf x_2'\mathbf x_2)^{-1} \mathbf x_2'$ is the projection onto $\mathbf x_2$.
Therefore, \eqref{me2} is identical to \eqref{e12} of the one-variable regression model. 
Then, using $\mathbf M_{2} \mathbf x_1\equiv \mathbf x$, we can apply the SIV method and determine an SIV $\mathbf s_1^\star$ that satisfies the DT condition. Next, repeating the same procedure for $\mathbf x_2$, we can also determine an SIV for this variable $\mathbf s_2^\star$.  These two SIVs allow us to estimate regression \eqref{me1}, using  $\mathbf Z=[\mathbf s_1^\star,\mathbf  s_2^\star]$, an $n \times k_z$ matrix of synthetic instruments, with $k_z = 2$. It straightforward to extend this procedure to cases with more than 2 endogenous variables. 

\section{\bf Sources of the data sets}

1. Mroz. dta: T.A. Mroz (1987), The Sensitivity of an Empirical Model of Married Women's Hours of Work to Economic and Statistical Assumptions, \textit{Econometrica} 55, 765--799.
\begin{verbatim}http://www.cengage.com/aise/economics/wooldridge_3e_datasets/\end{verbatim}

2. ipehd\_qje2009\_master.dta:  Becker and Woessmann (2009), Was Weber Wrong? A Human Capital Theory of Protestant Economic History, \textit{Quarterly Journal of Economics} 124 (2): 531--596. \begin{verbatim}https://www.ifo.de/sites/default/files/ipehd_qje2009_data_tables.zip\end{verbatim}

3. X401ksubs.dta: Introductory Econometrics: A Modern
Approach, Fifth Edition, Jeffrey M. Wooldridge. Source: Abadie, A. (2003). Semiparametric instrumental variable estimation of treatment
response models, \textit{Journal of Econometrics }113(2), 231--263.
\begin{verbatim}http://www.cengage.com/aise/economics/wooldridge_3e_datasets/\end{verbatim}
\subsection{Data Generating Process}\label{dgp}

We construct a DGP that mimics realistic features of economic data: non-normal distributions, heteroscedasticity, and substantial endogeneity bias. Following \citet{jonespewsey}, we use the sinh-arcsinh transformation to generate flexible distributional shapes.

\textbf{The sinh-arcsinh transformation.} Define the transformation
\[
H(\tilde{\mathbf x}; \epsilon, \kappa) = \sinh[\kappa \sinh^{-1}(\tilde{\mathbf x}) - \epsilon],
\]
where $\epsilon \in \mathbb{R}$ controls location/skewness and $\kappa \in \mathbb{R}_+$ controls kurtosis. Applying this to the normal CDF $\Phi(\cdot)$ yields the sinh-arcsinh family:
\[
\text{Skew-Normal}:=S(\tilde{\mathbf x}; \epsilon, \kappa) = \Phi[H(\tilde{\mathbf x}; \epsilon, \kappa)].
\]
When $\epsilon = 0$ and $\kappa = 1$, we recover the standard normal distribution. Different parameter values generate heavy-tailed, light-tailed, symmetric, and skewed distributions.

\textbf{Generating the endogenous regressor.} We create a sequence $\nu$ of $N$ equally-spaced points and transform via:
\[
\tilde{\mathbf x} = 7 \cdot H(\nu, 0, 0.5) + 1.1 \cdot \text{Uniform}(-1.01, 1.01).
\]
This produces a non-normal distribution with moderate kurtosis, mimicking observed patterns in economic variables.

\textbf{Constructing the endogenous error.} We generate the structural error $\mathbf u$ with two components:
\begin{enumerate}
\item \textbf{Endogenous component} (correlated with $\mathbf x$):
\[
e = (\tilde{\mathbf x} - \bar{\mathbf x}) + \varepsilon, \quad \varepsilon \sim N(0, \sigma_x^2).
\]

\item \textbf{Exogenous component} (uncorrelated with $\mathbf x$):
\[
u_1 = \text{Uniform}(-0.5, 0.5) + \text{Skew-Normal}(0, 5, 1.2),
\]
where we extract the residual $\mathbf v$ from regressing $\mathbf u_1$ on $\tilde{\mathbf x}$ and $\mathbf w$, then rescale:
\[
\mathbf v = \mathbf v \cdot \frac{\bar{\mathbf x}}{2}.
\]
\end{enumerate}

\textbf{Sign of endogeneity.} We construct the total error as: $\mathbf u = \mathbf e + \mathbf v $ assuming $\text{cov}(\mathbf x,\mathbf u) > 0$.
This allows us to test the sign determination procedure.

\textbf{Structural equation.} The outcome is generated as:
\[
\tilde{\mathbf y} = 1 + 2\tilde{\mathbf x} + 0.5\mathbf w + \mathbf u,
\]
where $\mathbf w \sim N(20, 10)$ is an exogenous control. The true causal effect is $\beta = 2$.

\textbf{Normalization.} We residualize both $\tilde{\mathbf y}$ and $\tilde{\mathbf x}$ on $\mathbf w$ to obtain $\mathbf y$ and $\mathbf x$, then standardize:
\[
\mathbf y = \frac{\mathbf y - \bar{\mathbf y}}{\sigma_y}, \quad \mathbf x = \frac{\mathbf x - \bar{\mathbf x}}{\sigma_x}.
\]

\textbf{Orthogonal vector.} We construct $\mathbf r$ as the residual from regressing $\mathbf y$ on $\mathbf x$, standardized to have unit variance. By construction, $\mathbf r \perp \mathbf x$.


\section{A sample code to implement the SIV method}\label{guide}
One can install the R-package and run the SIV method as follows.
{\begin{lstlisting}
remotes::install_git("https://github.com/ratbekd/siv.git")
library(siv)
## Example based on Mroz data
data <- wooldridge::mroz  # Use sample data set
data <- data[complete.cases(data), ]  # Remove missing values
# Run regression
#Y="hours" # outcome variable
#X="lwage"# endogenous variable
#H=c("educ", "age", "kidslt6", "kidsge6", "nwifeinc")# exogenous variables
result <- siv_reg(data, "hours", "lwage", c("educ", "age", "kidslt6", "kidsge6", "nwifeinc"), reps=5)
iv1 <- (result$IV1)
iv2 <-(result$IV2)
iv3 <-(result$IV3)
summ.iv1 <- summary(iv1, diagnostics=TRUE)
summ.iv2 <- summary(iv2, diagnostics=TRUE)
summ.iv3 <- summary(iv3, diagnostics=TRUE)

# In case of multiple endogenous variables use the following function
result <- msiv_reg(data, "hours", c("lwage", "educ"),c( "age", "kidslt6", "kidsge6", "nwifeinc"), reps=5)
iv1 <-(result$IV1)# a simple SIV
iv2 <-(result$IV2)# a robust parametric SIV (RSIV-p)
iv3 <-(result$IV3)# a robust non-parametric SIV (RSIV-n)
\end{lstlisting}}
\singlespacing

 \bibliographystyle{chicago}
 	\bibliography{iv}

\end{document}